%% file: Co_Ne_Sr_Distributed.tex
\documentclass[12pt,onecolumn,draftcls]{IEEEtran}
\usepackage{graphicx}
\usepackage{epsfig}
\usepackage{subfigure}
\usepackage{amssymb}
\usepackage{amsbsy}
\usepackage{amsmath}
\usepackage{cite}
\usepackage{url}

\usepackage{color}
\usepackage{float}
\usepackage{tabularx}
\usepackage{times,amsmath,epsfig}
\usepackage{xspace,latexsym,syntonly}
\usepackage{amssymb}
\usepackage{amsfonts}
\usepackage{textcomp}
\usepackage{subfigure}
\usepackage{amsbsy}
\usepackage{cite}
\input{macros}

\begin{document}
\title{Distributed Learning Algorithms for Spectrum Sharing in Spatial Random Access Wireless Networks}
\author{Kobi Cohen, Angelia Nedi\'c and R. Srikant
\thanks{Kobi Cohen is with the Department of Electrical and Computer Engineering, Ben-Gurion University of the Negev, Beer Sheva 8410501 Israel. Email:yakovsec@bgu.ac.il}
\thanks{Angelia Nedi\'c is with the with School of Electrical, Computer and Energy Engineering, Arizona State University, Tempe, AZ 85281, USA. Email:angelia.nedich@asu.edu}
\thanks{R. Srikant is with the Coordinated Science Laboratory and Department of Electrical and Computer Engineering, University of Illinois at Urbana-Champaign, IL 61801 USA. Email:rsrikant@illinois.edu}
\thanks{This work was supported by AFOSR MURI FA 9550–-10–-1–-0573, ONR Grant N00014–-13–-1–-003, NSF Grant CNS–-1161404, and ARO Grant W911NF--16--1--0259.}
\thanks{Part of this work was presented at the 13th International Symposium on Modeling and Optimization in Mobile, Ad Hoc and Wireless Networks (WiOpt), 2015 \cite{Cohen_2015_Distributed}.}
\thanks{Accepted for publication in the IEEE Transactions on Automatic Control.}
\thanks{Copyright (c) 2016 IEEE. Personal use of this material is permitted. Permission from IEEE must be obtained for all other
uses, in any current or future media, including reprinting/republishing this material for advertising
or promotional purposes, creating new collective works, for resale or redistribution to servers or lists,
or reuse of any copyrighted component of this work in other works.}
}
\date{}
\maketitle

\begin{abstract}
\label{sec:abstract}
We consider distributed optimization over orthogonal collision channels in spatial random access networks. Users are spatially distributed and each user is in the interference range of a few other users. Each user is allowed to transmit over a subset of the shared channels with a certain attempt probability. We study both the non-cooperative and cooperative settings. In the former, the goal of each user is to maximize its own rate irrespective of the utilities of other users. In the latter, the goal is to achieve proportionally fair rates among users. Simple distributed learning algorithms are developed to solve these problems. The efficiencies of the proposed algorithms are demonstrated via both theoretical analysis and simulation results.
\end{abstract}
%
\def\keywords{\vspace{.5em}
{\bfseries\textit{Index Terms}---\,\relax%
}}
\def\endkeywords{\par}
\keywords
Distributed optimization, collision channel, slotted-ALOHA, Nash equilibrium, proportional fairness.
\section{Introduction}
\label{sec:introduction}

Spectrum scarcity along with the increasing demand for wireless communication have triggered the development of efficient spectrum access schemes for wireless networks. A good overview of the various dynamic spectrum access models for MAC can be found in \cite{Zhao_Survey_2007}. In this paper we focus on the open sharing model among users that acts as the basis for managing a spectral region (e.g., WiFi, cognitive radio, sensor networks, and unlicensed band technology \cite{Zhao_Survey_2007}).

Consider a spatial wireless network with $N$ users sharing $K$ collision channels. Each user is in the interference range of a few (but not necessarily all) other users, referred to as neighbors (e.g., when the distance between users is small they cause mutual interference). In the beginning of each time slot, each user is allowed to transmit over $M$ channels ($1\leq M\leq K$) with a certain attempt probability (i.e., using the slotted-ALOHA protocol). If two or more neighbors transmit simultaneously over the same channel, a collision occurs. In multi-channel systems, exploiting the channel diversity plays an important role in designing effective channel allocation protocols. The channel conditions are a function of both the inherent quality of each channel due to fading, shadowing, etc., as well as the interference caused by the users that use the channel. Thus, it is intuitive that users can improve performance by adaptively choosing channels with a higher probability of being idle as well as higher capacity when idle. We are interested in finding a channel allocation and attempt probabilities in a distributed manner so as to optimize certain objectives in the network.

\subsection{Game theory, Distributed Optimization, and Learning for Spectrum Access Protocols}
\label{ssec:intro_game}

Spectrum access protocols can be broadly classified into two classes: (i) protocols in which users do not share information with each other, due to security or overhead considerations, and (ii) protocols in which information is shared to achieve a common goal, such as in networks which are controlled by a single provider. Achieving an effective channel allocation for the spectrum access problem in a distributed manner requires users to adaptively adjust their actions (i.e., select channels and attempt probabilities) based on local information about the current state of the system. Thus, the first question of interest is whether the system keeps oscillating due to frequent channel switching, or whether the system converges to a stable operating point. When users do not share information, a stable channel allocation may not be a system-wide optimal solution (though it reduces the undesirable effects of frequent channel switching and also demonstrated good performance in some network models and typical scenarios, as in \cite{Yu_2002_Distributed, Leshem_Multichannel_2012, Cohen_Game_2013}). Thus, the second question of interest is whether small amounts of information sharing can lead to a globally-optimal operating point.

Game theory provides a rich set of tools to analyze a dynamic behavior of a system when entities (or players) in the system take actions to optimize a predefined objective. Thus, using game theoretic models to analyze wireless networking protocols and algorithms, in which users (i.e., players) adjust their strategies (e.g., attempt probabilities, transmission power, selected channels, etc.) so as to optimize the system performance has been attracted much attention in recent years. Related work on networking games can be found under a non-cooperative setting in \cite{Mackenzie_Selfish_2001, Jin_Equilibria_2002, Alpcan_2002_CDMA, Han_2005_Fair, Altman_2006_Survey, Leshem_2006_Brgaining, Ozdaglar_2007_Incentives, Boche_2007_Non, Nokleby_2007_Cooperative, Menache_Rate_2008, Inaltekin_2008_Analysis, Jorswieck_2008_Miso, Gao_2008_Game, Leshem_2008_Cooperative, Candogan_Competitive_2009, Altman_2009_Operating, Menache_2011_Network, cohen2012game, Cohen_Game_2013, Cohen_Distributed_2013} and under a cooperative setting in  \cite{Leshem_2006_Brgaining, Boche_2007_Non, Jorswieck_2008_Miso, Gao_2008_Game, Nokleby_2007_Cooperative, Han_2005_Fair, Leshem_2008_Cooperative, cohen2013distributed, Cohen_Distributed_2013}. Since generally we are interested in networking protocols that require small amounts of information sharing (if any) and distributed in nature, it is often desired to develop efficient distributed learning and optimization methods to achieve the target solution.

\subsection{Main Results}
\label{ssec:main}

We first examine the case where users do not share information with each other. The achievable rate of each user increases with its own attempt probability, when other attempt probabilities are fixed. Thus, a natural approach to achieve a good operating point is to allow every user to maximize its own rate under a constraint on the allowed attempt probability\footnote{Similar approaches were applied in \cite{Menache_Rate_2008, Jin_Equilibria_2002, Cohen_Game_2013}, all resulting with an individual rate and attempt probability for every user.} (where different attempt probability constraints are used to prioritize different users in the network), referred to as distributed rate maximization. Next, we summarize our main results in this respect. (i) In \cite{Cohen_Game_2013}, the special case of a fully connected network (i.e., all users are in the same interference range) and $M=1$ (i.e., each user is allowed to transmit over only one channel) was considered, and a distributed algorithm was applied to solve the distributed rate maximization problem, in which each user updates its strategy using its local channel state information (CSI) and by monitoring the load over the channels. It was shown that any finite improvement path (not necessarily best-response) across users, in which at each iteration the rate of a single user increases when it updates its channel-selection strategy given the current system state, reaches an equilibrium in the sense that no user can increase its rate by unilaterally changing its strategy. In this paper, however, we consider a more general case where each user interferes only with its neighbors, and $M\geq 1$ (i.e., every user is allowed to transmit over multiple channels). Interestingly, we show that cycles may occur under some improvement paths in this general model. To solve this problem, we use the theory of \emph{best-response (BR) potential games}, introduced by Voorneveld in 2000 \cite{Voorneveld_2000_Best}. In BR potential games, cycles may occur under some improvement paths, though no cycles occur under a BR dynamics. We prove that the system dynamics can be formulated as a BR potential game. This result constitutes an important contribution from a game theoretic perspective as well as MAC design perspective, since it generalizes existing results on Nash equilibria (NE) in \cite{Mavronicolas_Congestion_2007, Chen_2013_Distributed, Cohen_Game_2013} (see a more detailed discussion in Section \ref{ssec:related}). (ii) Based on our analysis, we then propose a distributed BR learning algorithm that solves the distributed rate maximization problem and converges to an equilibrium in finite time. The convergence result described above requires a coordination mechanism that enables users to update their actions sequentially. We then propose a simpler mechanism that guarantees convergence as time increases even without coordination among users (thus, users may update actions simultaneously). We further extend our convergence result to cases where each user may have a different set of available resources, which captures the situation of a hierarchical model as in cognitive radio networks (see Section \ref{sec:rate_maximization} for details). Thus, these results enable us to design MAC protocols for a wide range of practical system models. (iii) Since multiple NEs may exist, we finally analyze the efficiency of the NEs that the algorithm may converge to. It should be noted that very little is known about the efficiency of the NEs under related models as considered in this paper, particularly when interference across users forms a graph structure. A popular performance measure for a NE efficiency is the Price of Anarchy (PoA), which is the ratio between the optimal performance and the worst equilibrium. The PoA (with respect to the sum utility) has been analyzed in \cite{law2012price} under the special case of a fully connected network (i.e., the interference graph is complete) and equal attempt probability for all users. In this paper, we analyze performance at equilibrium on average rather than worst case performance, which is useful particularly in the context of wireless networks since we are generally interested in the expected performance of users in the long run. Specifically, it is shown that under some mild conditions (see Section \ref{ssec:efficiency} for details), implementing the proposed algorithm on regular conflict graphs guarantees that every user in the system improves its performance (in terms of expected rate) as compared to a naive algorithm, in which users choose channels for transmission randomly without performing congestion monitoring used to adjust their strategies as proposed in this paper. Significant performance gain (more that $170\%$ improvement) is obtained under a low collision level.

Second, we focus on a cooperative setting, in which the goal is to achieve the optimal channel allocation and attempt probabilities that attain proportionally fair rates in the network. When $K=1$ (i.e., a single channel case), users have no freedom to choose among different channels, and the action of each user degenerates to setting the optimal attempt probability for transmission over the single channel. Low-complexity algorithms have been developed in \cite{Kar_2004_Achieving, Wang_2006_Cross, Gupta_2012_Throughput} under various models of a single collision channel. In this paper, however, we address this question for multi-channel networks (i.e., $K\geq 1$) where every user is allowed to choose a single channel for transmission (i.e., $M=1$) among the $K$ channels and to set the optimal attempt probability for transmission over the channel\footnote{Accessing a single-channels is often assumed due to hardware constraints or when it is desired to limit the congestion level in high-loaded systems. It has been widely assumed in cognitive radio applications, WiFi, sensor networks, etc. It should be noted, however, that developing a tractable optimal solution for the proportional fairness problem under the case where users are allowed to access two or more channels at a time remains an open question.}. Direct computation of the optimal channel allocation and attempt probabilities that attain proportionally fair rates for the multi-channel ALOHA network considered in this paper is a combinatorial optimization problem over a graph. Furthermore, it requires a centralized solution that uses global information which is impractical in large-scale networks. Next, we summarize our main results in this respect. (i) We study the problem from a game theoretic perspective and develop a novel cooperative distributed algorithm based on log-linear learning, referred to as noisy BR, to achieve the target solution in a distributed manner. Specifically, at each iteration, using message exchanges between neighbors only, selected users take actions with respect to a cooperative utility that balances between their own utilities and the interference level they cause to their neighbors given the current system state. In noisy BR dynamics, users play the BR that maximizes their cooperative utilities with high probability, while suboptimal responses are taken with small probabilities to escape local maxima. We prove that the proposed cooperative algorithm converges to the global proportional fairness solution with high probability as time increases. Furthermore, we show that every Nash equilibrium attained by the algorithm can be reached in finite time by playing BR and it is a good operating point in the sense that proportionally fair rates are attained locally among all users sharing the same channel. (ii) The proposed algorithm significantly simplifies the implementation as compared to existing methods. First, it requires less amount of information sharing between nodes. Second, synchronization in a neighborhood with respect to action updates is not required (see Section \ref{ssec:related} for a more detailed discussion on related works).

\subsection{Related Work}
\label{ssec:related}

Spectrum access and sharing have attracted much attention in past and recent years. We next discuss related works that use game theoretic models, distributed optimization, and learning techniques, some of them have been discussed in Sections \ref{ssec:intro_game}, \ref{ssec:main}, and highlight the main differences in the model, analysis and results obtained in this paper as compared to the related existing studies.

\textbf{ALOHA-based Protocols and Cross Layer Optimization}. ALOHA-based protocols have been widely used in wireless communication primarily because of their ease of implementation and their random nature. Related work on ALOHA-based protocols can be found in \cite{Pountourakis_Analysis_1992, Mackenzie_Selfish_2001, Jin_Equilibria_2002, Shen_Stabilized_2002, Altman_2004_Slotted, Bai_Opportunistic_2006, Menache_Rate_2008, To_2010_Exploiting, Cohen_Game_2013, Cohen_Distributed_2013, Wu_2013_Fasa} for fully connected networks and in \cite{Kar_2004_Achieving, Wang_2006_Cross, Baccelli_2006_Aloha, kauffmann2007measurement, Baccelli_2009_Stochastic, Gupta_2012_Throughput, Chen_2013_Distributed, Hou_2014_Proportionally} for spatial networks. Stability of a selfish behavior dynamics in a single-channel ALOHA system was studied in \cite{Mackenzie_Selfish_2001}. Equilibria under rate demands have been analyzed in \cite{Jin_Equilibria_2002, Menache_Rate_2008}. In this paper, however, we focus on the multi-channel case. In \cite{Chen_2013_Distributed, Cohen_Game_2013}, the multi-channel ALOHA case was studied, where $M=1$. In \cite{Chen_2013_Distributed}, the authors have developed a distributed algorithm, in which a mixed strategy was applied to obtain local information in a spatially distributed network. In \cite{Cohen_Game_2013}, a pure strategy was applied, where the local information was obtained by sensing the spectrum in a fully connected network. When $M=1$, the log-rate of each user under an ALOHA model can be expressed as a linear combination of its inherent log-rate minus the log-interference caused by its neighbors (i.e., an affine function, see (\ref{eq:log_rate}) for details). As a result, due to the monotonicity of the logarithm, analysis of Nash equilibria when $M=1$ under the non-cooperative setting follows by applying a variation of the ordinal potential function introduced in \cite{Mavronicolas_Congestion_2007} for affine utilities. Thus, any improvement path (not necessarily best-response) across users, in which at each iteration the rate of a user increases when it updates its channel-selection strategy given the current system state (i.e., sequential updating), reaches an equilibrium in the sense that no user can increase its rate by unilaterally changing its strategy. In this paper, however, we consider the case where $M\geq 1$, in which cycles may occur under some improvement paths and the dynamics does not obey an ordinal potential function. Thus, our Nash equilibria analysis using the theory of best-response potential games (as described in Section \ref{ssec:main}) generalizes the Nash equilibria results obtained in \cite{Mavronicolas_Congestion_2007, Chen_2013_Distributed, Cohen_Game_2013}. It also generalizes the equilibria results in \cite{Southwell_2012_Convergence} (that assumes that each node contributes equally to the congestion of a resource) due to different attempt probabilities across users considered here. It should be noted that avoiding simultaneous updates across users can be done by allowing each user to draw a random backoff time and update its strategy when the backoff time expires. However, we will show convergence of the algorithm even without this mechanism. Stability of multi-channel ALOHA systems was studied in \cite{Saadawi_1981_Aanalysis, Ghez_1988_Stability, Pountourakis_Analysis_1992, Shen_Stabilized_2002}. In \cite{Baccelli_2006_Aloha, Baccelli_2009_Stochastic}, spatial single-channel ALOHA networks have been studied under interference channels using stochastic geometry. Opportunistic ALOHA schemes that use cross layer MAC/PHY techniques, in which the design of Medium Access Control (MAC) is integrated with physical layer channel information to improve the spectral efficiency, have been studied under both the single-channel \cite{Bai_Opportunistic_2006, Menache_Rate_2008, Baccelli_2009_Stochastic} and multi-channel \cite{Bai_Opportunistic_2006, To_2010_Exploiting, Cohen_Game_2013, Cohen_Distributed_2013} cases. Other related studies considered recently opportunistic carrier sensing in a cross-layer design \cite{cohen2009time, cohen2010time, cohen2010likelihood, cohen2011energy, Leshem_Multichannel_2012}. A cross-layer MAC/PHY methodology is used in this paper to design efficient distributed algorithms for the problems under study.

\textbf{Distributed Learning and Optimization for a Fair Spectrum Sharing.} Achieving proportionally fair rates in spatial random access networks (which considered in this paper under the cooperative setting as described in Section \ref{ssec:main}) has been studied under the single collision channel case ($K=1$) in \cite{Kar_2004_Achieving, Wang_2006_Cross, Gupta_2012_Throughput} and the multi collision channel case ($K\geq 1$, $M=1$) in \cite{Hou_2014_Proportionally} (as considered in this paper). The algorithm developed in \cite{Hou_2014_Proportionally} uses a Gibbs sampler over local maxima that converges to a global maximum as time increases. The algorithm requires information sharing between nodes up to second neighborhood at each iteration. It further requires perfect synchronization in a neighborhood with respect to action updates in the sense that once a node updates its strategy all its neighbors must update their strategies accordingly. In this paper, however, we develop an algorithm that requires information sharing between a single node and its neighbors only (i.e., first neighborhood) at each iteration, and synchronization in a neighborhood with respect to action updates is not required. Once a node updates its strategy, its neighbors may or may not update their strategies. Thus, convergence of our scheme is robust against stubborn neighbors, temporary communication link failures, etc. The proposed algorithm is based on log-linear learning techniques (see \cite{Young_1998_Individual, Marden_2012_Revisiting} for more details on the theory of log-linear learning), and use a game theoretic perspective to analyze the algorithm's performance. Similar idea for using altruistic plus selfish components in the algorithm design under a channel and cell selection problem has been identified in \cite{kauffmann2007measurement}, where global optimum was obtained via Gibbs sampler. The MAC layer protocol between users was assumed given and the question of interest is concerned with the interference mitigation between cells which are co-exist in the same frequency bands. Furthermore, the objective aimed at minimizing the minimum potential delay (and not obtaining proportionally fair rates as considered in this paper). On algorithm development, the model in \cite{kauffmann2007measurement} requires each user to computes the aggregate utility of its own and all users in the network that communicate with the same AP (via a utility of the form of $1/f(SNR)$). Consequently, the resulting algorithm in \cite{kauffmann2007measurement} is fundamentally different from the one developed in this paper under the cooperative setting. Other related studies that use log-linear learning and Gibbs sampling techniques under different spectrum access models and objectives can be found in \cite{Xu_2012_Opportunistic, Herzen_2013_Distributed, Singh_2013_Distributed, Jang_2014_Distributed, Herzen_2015_Learning}.

\textbf{Game Theoretic Models for Communication Systems}. Cooperative game theoretic optimization has been studied under frequency flat interference channels in the SISO \cite{Leshem_2006_Brgaining, Boche_2007_Non}, MISO \cite{Jorswieck_2008_Miso, Gao_2008_Game} and MIMO cases \cite{Nokleby_2007_Cooperative}. The frequency selective interference channels case has been studied in \cite{Han_2005_Fair, Leshem_2008_Cooperative}. The collision channels case has been studied under a fully-connected network and without information sharing between users in \cite{Cohen_Distributed_2013}, where the global optimum was attained under the asymptotic regime (i.e., as the number of users $N$ approaches infinity) and the i.i.d assumption on the channel quality. In this paper, however, we study distributed optimization of the user rates under the cooperative setting for spatial networks where information sharing between neighbors is allowed. We show that proportionally fair rates are attained for any number $N\ge 1$ of users without any assumption on the network topology or channel distribution.

Other related game theoretic models have been used in cellular, OFDMA, and 5G systems \cite{scutari2006potential, el2012stable, zhao2014coordinated, zhang2014distributed, wang2016dense}. In \cite{scutari2006potential}, the authors focused on a power control model, where exact and ordinal potential game models have been investigated. In \cite{el2012stable}, a joint uplink/downlink subcarrier allocation in OFDMA systems has been investigated via a two-sided stable matching game formulation. In \cite{zhao2014coordinated}, the interference mitigation problem in the downlink of multicell networks via base station coordination has been studied via a potential game framework. In \cite{zhang2014distributed}, the authors investigated channel utilization via a distributed matching approach. In \cite{wang2016dense}, a distributed power self-optimization problem has been studied for the downlink operation of dense femtocell networks via a noncooperative exact potential game formulation. This paper, however, considers a fundamentally different model, where communication is over collision channels (i.e., interferences are caused by the MAC layer's attempt probabilities), and the optimization variables are channel allocation and attempt probabilities. From a game theoretic perspective, we show that some improvement paths may result in cycles under the noncooperative setting (thus, the game dynamics does not obey exact or ordinal potential functions). Instead, we formulate the game as a best-response potential game, where it is shown that best-response dynamics converges.

\textbf{Spectrum Access as a Graph Coloring Problem}. Another set of related works is concerned with modeling the spectrum access problem as a graph coloring problem, in which users and channels are represented by vertices and colors, respectively. Thus, coloring vertices such that two adjacent vertices do not share the same color is equivalent to allocating channels such that interference between neighbors is being avoided (see \cite{Wang_2005_List, Wang_2009_Improved, Checco_2013_Learning, Checco_2014_Fast} and references therein for related works). However, the problem considered in this paper is different since we mainly focus on the case where the number of users is much larger than the number of channels (thus, coloring the graph may be infeasible). Furthermore, in our case users may select more than one channel, and may prefer some channels over others, as well as optimize their rates with respect to the attempt probability.

The rest of the paper is organized as follows. In Section~\ref{sec:network} we describe the network model. In Sections~\ref{sec:rate_maximization} and \ref{sec:fairness} we consider the noncooperative and cooperative settings, respectively. In Section~\ref{sec:simulation} we provide simulation results. Section~\ref{sec:conclusion} concludes the paper.

\section{Network Model}
\label{sec:network}

We consider a wireless network consisting of a set $\mathcal{N}=\left\{1, 2, ..., N\right\}$ of users (or transceiver links) and a set of $\mathcal{K}=\left\{1, 2, ..., K\right\}$ of shared channels (where typically $N>K$). We focus on a spatial wireless network, where each user is in the interference range of a few (but not necessarily all) other users. We assume symmetric interference ranges for all users in the sense that user $n$ is in user $r$'s interference range only if user $r$ is in user $n$'s interference range for all $n, r\in\mathcal{N}$. We refer to users in the same interference range as \emph{neighbors}, and define $\mathcal{I}_n\subseteq(\mathcal{N}\setminus n)$ as the set of user $n$'s neighbors (i.e., the interference range equals the communication range when considering communication between neighbors). We assume that users are backlogged, i.e., all $N$ users always have packets to transmit. In the beginning of each time slot, each user (say $n$) is allowed to transmit over $M$ channels ($1\leq M\leq K$) with a certain attempt probability (i.e., using the slotted-ALOHA protocol). Let $\mathcal{K}_M$ be the set of all $M$-element subsets of $\mathcal{K}$ (i.e., $\mathcal{K}_M$ is the set of all channel-selection strategies that a user can choose). Let $\sigma_n=(k_n,p_n)$ be the strategy of user $n$, where $k_n=\left\{k_{n,i}\right\}_{i=1}^M\in\mathcal{K}_M$ denotes the set of chosen channels and $0\leq p_n\leq 1$ denotes the attempt probability of user $n$. Thus, when user $n$ decides to transmit (which occurs with probability $p_n$) it uses all the channels in $k_n$ for transmission. We define $\sigma$ as the strategy profile for all users, and $\sigma_{-n}$ as the strategy profile for all users except user $n$.

The topology of the interference model can be represented by an undirected graph $G= (\mathcal{N}, E)$, where the set of users are represented by the vertices and the interference relationships between users are represented by the set of edges $E$. An edge $(n, r)\in E$ means that users $n$ and $r$ are in the same interference range. The set of user $n$'s neighbors $\mathcal{I}_n$ is represented by vertices directly connected to vertex $n$ excluding vertex $n$ itself. An illustration is given in Fig. \ref{fig:small_network} in Section \ref{sec:simulation}.

We consider transmissions over orthogonal collision channels. Thus, transmission by user $n$ over channel $k_{n,i}$ is successful only if no user $r\in\mathcal{I}_n$ transmits over channel $k_{n,i}$ in the same time-slot. However, if user $n$ and at least one more user in $\mathcal{I}_n$ transmit simultaneously over channel $k_{n,i}$ in the same time slot, a collision occurs. The achievable rate of user $n$ over channel $k$ given that a transmission is successful, referred to as collision-free utility, is denoted by $u_n(k)\geq 0$ (i.e., Shannon capacity). We consider long-term rates where $u_n(k)$ remain fixed across time slots during the running-time of the algorithms (e.g., mean-rate, or slow-fading effect). It should be noted that the algorithm dynamics and convergence analysis
hold under any network topology and when rates (i.e., channel gains) may be different across users and frequencies. However, equal channels are required for purposes of analysis in Section \ref{ssec:efficiency}.

Define the success probability of user $n$ on channel $k$, given the strategy profile of other users, as follows:
\beq
\label{eq:v}
\bea{l}
\displaystyle v_{n}(k,\sigma_{-n})\triangleq \prod_{i\in\mathcal{I}_n}{\left(1-p_i\right)^{\mathbf{1}_i(k)}}
  \;,
\ena
\eeq
where ${\mathbf{1}_i(k)}=1$ if $k\in k_i$ and ${\mathbf{1}_i(k)}=0$ otherwise. Hence, the expected rate of user $n$ over channel $k_{n,i}$ is given by:
\beq
\label{eq:R}
\bea{l}
\displaystyle r_n\left(k_{n,i}, p_n, \sigma_{-n}\right)
=p_n u_n(k_{n,i})v_n(k_{n,i},\sigma_{-n})\;.
\ena
\eeq
Note that the log-rate of user $n$ over channel $k_{n,i}$ is given by
\beq
\label{eq:log_rate}
\displaystyle\log r_n\left(k_{n,i}, p_n, \sigma_{-n}\right)
=\log \left(u_n(k_{n,i})p_n\right)-I_n(k_{n,i},\sigma_{-n}),
\eeq
where $I_n(k,\sigma_{-n})$ is referred to as the \emph{log-interference} function and is given by:
\beq
\label{eq:log_interference}
\displaystyle
I_n(k,\sigma_{-n})\triangleq-\log v_{n}(k,\sigma_{-n})=\sum_{i\in\mathcal{I}_n}{\log\left(\frac{1}{1-p_i}\right)\mathbf{1}_i(k)}\;.
\eeq
Note that $I_n(k,\sigma_{-n})$ can be viewed as the log-interference that user $n$ experiences over channel $k$ caused by its neighbors that transmit over the same channel. Finally, the expected rate of user $n$ is given by:
\beq
\label{eq:R2}
\bea{l}
\displaystyle 
R_n\left(\sigma\right)
\triangleq\sum_{i=1}^{M} r_n\left(k_{n,i}, p_n, \sigma_{-n}\right)\;.
\ena
\eeq

Throughout the paper, we will develop distributed algorithms to optimize certain objectives in the network. Theoretically, convergence analysis often requires users to update their strategies in a sequential manner. Avoiding simultaneous updates in communication systems is often done by allowing each user to draw a random backoff time and update its strategy when its backoff time expires (as discussed in Section \ref{ssec:related}). For simplicity, we will assume a similar mechanism here. Specifically, it is assumed that users hold a global clock and may update their strategies only at times $t_1, t_2, ...$, referred to as \emph{updating times}. At each updating time, every user draws a backoff time from a continuous uniform distribution over the range $[0, B]$ for some $B>0$. A user whose backoff time expires may broadcast a pilot signal to its neighbors, indicating that its strategy has been updated or start transmitting its data and its neighbors can sense activity. Then, all its neighbors keep their strategies fixed until the next updating time. Note that neighbors will not update their strategies simultaneously, and the time interval for data transmissions is set to be higher than $B$. At each updating time, we refer to users that update their strategies as \emph{active users}. The set of active users is denoted by $\mathcal{N}_a$ (which is time-varying across updating times). In Tables \ref{tab:BR_DRM}, \ref{tab:NBRF} (Step $3$) we refer to this mechanism as a selection of active users. It should be noted, however, that convergence of the algorithm discussed in Section \ref{ssec:BR_DRM} will be shown even without this coordination mechanism.

\section{Distributed Rate Maximization: \\ A Non-Cooperative Setting}
\label{sec:rate_maximization}

In this section we consider the case where every user (say $n$) maximizes its own rate given the current system state under a constraint $P_n$  on its allowed attempt probability, i.e., $p_n\leq P_n$ where $P_n<1$ (see Section \ref{ssec:main} for motivation of this problem). Since maximizing the rate given the current system state results in a transmission with the maximal allowed attempt probability $P_n$, the strategy for user $n$ degenerates to choosing the subset of channels $k_n$ that maximizes its own rate under a fixed attempt probability $P_n$. As a result, the strategy played by user $n$ given a fixed strategy profile of other users $\sigma_{-n}$ is given by $\sigma_n=(k_n^*, P_n)$, where $k_n^*=\left\{k_{n,i}^*\right\}_{i=1}^{M}$ solves the following distributed rate maximization problem\footnote{For ease of presentation, we assume continuous random rates $u_n(k)$ to guarantee a uniqueness of the maximizer. Otherwise, channels with the same rate can be ordered arbitrarily.}:
\beq
\label{eq:pre_non_cooperative_optimization}
\bea{l}
\displaystyle k_n^*=\arg\;\max_{k_n\in\mathcal{K}_M}\hspace{0.3cm} R_n\left(\sigma\right) \hspace{0.3cm} \mbox{s.t.} \hspace{0.3cm} \displaystyle p_n=P_n
  \;.
\ena
\eeq
Since $R_n\left(\sigma\right)=p_n\sum_{i=1}^M u_n(k_{n,i})v_n(k_{n,i},\sigma_{-n})$ and $p_n=P_n$ in (\ref{eq:pre_non_cooperative_optimization}) is a constant independent of $k_n$, it suffices to solve:
\beq
\label{eq:non_cooperative_optimization}
\bea{l}
\displaystyle k_n^*=\arg\;\max_{k_n\in\mathcal{K}_M}\hspace{0.3cm} \sum_{i=1}^M u_n(k_{n,i}) v_n(k_{n,i},\sigma_{-n})
  \;.
\ena
\eeq
For every user $n$ let $\left\{k_{n,1}^*, k_{n,2}^*, ..., k_{n,K}^*\right\}$ be a permutation of $\left\{1, ..., K\right\}$ such that:
\beq
\label{eq:BR_DRM_sol2}
\bea{l}
u_n(k_{n,1}^*)v_n(k_{n,1}^*,\sigma_{-n})\geq u_n(k_{n,2}^*)v_n(k_{n,2}^*,\sigma_{-n}) \vspace{0.0cm} \\ \hspace{3cm}
\geq\cdots\geq u_n(k_{n,K}^*)v_n(k_{n,K}^*,\sigma_{-n})\;.
\ena
\eeq
Following (\ref{eq:non_cooperative_optimization}), the channel-selection strategy that solves (\ref{eq:pre_non_cooperative_optimization}) at each given updating time is given by:
\beq
\label{eq:BR_DRM_sol3}
k_n^*=\left\{k_{n,1}^*, k_{n,2}^*, ..., k_{n,M}^*\right\}\;.
\eeq

Note that in practical systems, user $n$ holds an estimate of $u_n(k)$ (from pilot signals for instance). On the other hand, complete information about other user strategies is not required. Monitoring the channels to obtain $v_n(k,\sigma_{-n})$ for all $k$ is sufficient to make a decision\footnote{Note that the number of idle time slots and busy time slots can be used to estimate the success probability. Monitoring the channels can be done by the receiver (which can sense the spectrum and send this information to the transmitter). Another way is to monitor the null period by the transmitter as in cognitive radio systems. Any attempt to access channel $k$ by one user or more results in identifying channel $k$ as busy.}. Hence, for purposes of analysis in this section we assume that every user $n$ estimates $v_n(k,\sigma_{-n})$ perfectly (i.e., monitors the channels for a sufficient time). In Section \ref{sec:simulation}, simulation results demonstrate strong performance of the proposed algorithm in practical systems under estimation errors. Next, we examine a distributed algorithm that uses $u_n(k)$, $v_n(k,\sigma_{-n})$ to solve the distributed rate maximization problem.

\subsection{Best-Response Potential Game Formulation}
\label{ssec:BR_potential_game}

The system dynamics can be viewed as a non-cooperative game, in which every user sequentially updates its strategy to increase its rate given the current system state irrespective of other users' rates, referred to as the \emph{Distributed Rate Maximization (DRM) game}. The strategy $k_n^*$ that solves (\ref{eq:pre_non_cooperative_optimization}) represents a \emph{best-response (BR)} strategy since a user chooses $k_n^*$ that maximizes its rate given the current system state. On the other hand, switching from strategy $k_n$ to $k_n'$ to increase the rate (but not maximizing it) such that $R_n(k_n', P_n, \sigma_{-n})>R_n(k_n, P_n, \sigma_{-n})$ is called a \emph{better-response}. A system is in an equilibrium when users cannot increase their rates by unilaterally changing their strategy.
\begin{definition}
A Nash Equilibrium Point (NEP) for the DRM game is a strategy profile $\sigma^*=(\sigma_n^*,\sigma_{-n}^*)$, where $k_{n'}^*\in\mathcal{K}_M$, $p_{n'}^*=P_{n'}$ for all $n'\in\mathcal{N}$, such that
\beq
\label{eq:NEP}
\bea{l}
\displaystyle R_n\left(\sigma_n^*,\sigma_{-n}^*\right)\geq
\displaystyle R_n(\tilde{\sigma}_n,\sigma_{-n}^*)  \vspace{0.0cm}\\\hspace{2cm} \forall n \;, \forall \tilde{\sigma}_n=(\tilde{k}_n, P_n)\;,\; \tilde{k}_n\in\mathcal{K}_M
  \;.
\ena
\eeq
\end{definition}

A game has the \emph{finite improvement property (FIP)} if every \emph{improvement path}, in which a sequence of better-responses are executed by users sequentially, is finite. Clearly, a game with FIP converges to a NEP in finite time under any better-response dynamics. In what follows we use the theory of potential games to analyze the convergence of the BR dynamics to a NEP under the DRM game. In potential games, the incentive of users to switch strategies can be expressed by a global potential function. A NEP for the game is reached at any local maximum of the potential function. Next, we define a class of related potential games to the DRM game at hand.

\begin{definition}[{\cite{Voorneveld_2000_Best}}]
The DRM game is referred to as a \emph{best-response potential game} if there is a best-response potential function $\phi : \sigma\rightarrow \mathbb{R}$ such that for every user $n$ and for every $\sigma_{-n}=\left\{k_i, p_i\right\}_{i\neq n}$, where $k_i\in\mathcal{K}_M$, $p_i=P_i$, the following holds:
\beq
\label{eq:ordinal_potential_def}
\bea{l}
\displaystyle \arg\max_{k_n\in\mathcal{K}_M} R_n(k_n, P_n, \sigma_{-n})
=\arg\max_{k_n\in\mathcal{K}_M}\phi(k_n, P_n, \sigma_{-n}) \;.
\ena
\eeq
\end{definition}
Differing from other classes of potential games (e.g., exact, ordinal) which have the FIP, cycles may occur in BR potential games under some improvement paths. Nevertheless, no cycle occurs when playing BR dynamics since the potential function increases at any BR. In the DRM game, some improvement paths may result in cycles when $M>1$, as shown in Appendix \ref{app:example}. Nevertheless, the following theorem shows that the DRM game is a best-response potential game.
\textsl{\theorem\label{th:BR_potential_game}{
The DRM game is a best-response potential game, with the following best-response potential function:
\beq\label{eq:BR_potential_function}
\bea{l}
\phi(\sigma) = \displaystyle\sum_{n=1}^{N}\log\left(\frac{1}{1-P_n}\right)
\times\vspace{0.0cm}\\ \hspace{2cm}
\displaystyle\sum_{i=1}^{M}\left(\log u_{n}(k_{n,i})-\frac{I_n(k_{n,i},\sigma_{-n})}{2} \right)
\;.
\ena
\eeq
}}\vspace{0.0cm} \\
\begin{proof}
The proof is given in Appendix \ref{app:pr_ordinal}. \vspace{0.0cm}
\end{proof}
Note that a variation of (\ref{eq:BR_potential_function}) was shown to be an ordinal potential function for a game with affine utilities in \cite{Mavronicolas_Congestion_2007} (i.e., any improvement path reaches an equilibrium in finite time). Theorem \ref{th:BR_potential_game}, however, shows that a best-response dynamics under the DRM game reaches an equilibrium in finite time although cycles may occur under some improvement paths.

\begin{remark}
It should be noted that when the constraints on the attempt probabilities satisfies: $P_n(k)\in\left\{0, P_n\right\}$ for all $k, n$, each user selects channels among the set of channels, in which $P_n(k)=P_n>0$. Thus, it can be verified that Theorem \ref{th:BR_potential_game} holds under this more general case as well. This scenario captures the situation of a hierarchical model (as in cognitive radio networks). An example of such attempt probability constraints is depicted in Fig. \ref{fig:mask_constraint}, where user $1$ (high-priority) is allowed to transmit over white spaces and $2.4$GHz bands, while user $2$ (low-priority) is allowed to transmit over $2.4$GHz band only.
\end{remark}

\begin{figure}[htbp]
\centering \epsfig{file=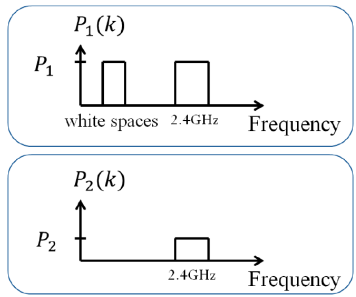,
width=0.35\textwidth}
\caption{An illustration of attempt probability constraints in a scenario of a hierarchical model in cognitive radio networks. User $1$ (high-priority) is allowed to transmit over white spaces and $2.4$GHz bands, while user $2$ (low-priority) is allowed to transmit over $2.4$GHz band only.}
\label{fig:mask_constraint}
\end{figure}

\subsection{Best-Response Algorithm for Distributed Rate Maximization}
\label{ssec:BR_DRM}

Following Theorem \ref{th:BR_potential_game}, we propose a non-cooperative BR algorithm to solve the constrained distributed rate maximization problem in the spatial multi-channel ALOHA networks, dubbed BR for Distributed Rate Maximization (BR-DRM) algorithm. We initialize the algorithm by a simple solution where every user picks the $M$ channels with the highest collision-free utility $u_n(k)$. In the learning process step, each user monitors the load on the channels to obtain $v_n(k,\sigma_{-n})$ for all $k$ (see the beginning of Section \ref{sec:rate_maximization} for more details on the monitoring process). Then, at each updating time the selected active users (selected according to the mechanism described in Section \ref{sec:network}) update their strategies by selecting the channels according to (\ref{eq:BR_DRM_sol3}). When users cannot increase their rates by unilaterally changing their strategy, an equilibrium is obtained. The BR-DRM Algorithm is given in Table \ref{tab:BR_DRM}. The set of active users in Step 3 is determined according to the distributed mechanism described in Sec. \ref{sec:network}. In Steps $5-7$ the user selects the channels for transmission based on the estimated load. Users repeat updating strategies until their rates converge. During the running time of the algorithm the loads on the channels are changed dynamically and affect user decisions across time. Convergence is guaranteed following Theorem \ref{th:BR_potential_game}, since the best response potential function is upper bounded (by $\phi(\sigma)\leq M\sum_{n=1}^{N}{\log\left(\frac{1}{1-P_n}\right)\max_{k}\log\left({u}_n(k)\right)}$) and any local maxima is a NEP for the game (since no user can increase its rate by unilaterally changing its strategy). It should be noted that convergence in finite time of BR dynamics in the DRM game is preserved as long as all active users are not neighbors (since the log-interference that user $n$ experiences $I_n(k,\sigma_{-n})$ is affected only by users in $\mathcal{I}_n$, thus we assume that no simultaneous updates occur among neighbors) as designed by the mechanism that selects the active users described in Section \ref{sec:network}.
\begin{corollary}
Assume that users update their strategy according to the mechanism described in Section \ref{sec:network}. Then, the BR-DRM algorithm, given in Table \ref{tab:BR_DRM}, converges to a NEP in finite time. \vspace{0.0cm}
\end{corollary}

Next, we examine the case where simultaneous updates across neighbors may occur. An example for this case is when a simpler mechanism is applied where each user (say $n$) updates its strategy with a given probability $0<q_n<1$, referred to as a \emph{probabilistic mechanism}. Another example is when communication errors between neighbors result in simultaneous updates. In such cases, convergence to a NEP is achieved with high probability as time increases. \vspace{0.0cm}

\begin{proposition}
\label{prop:probabilistic}
Assume that users update their strategy according to the probabilistic mechanism. Then, the BR-DRM algorithm converges to a NEP with probability $1$ as time approaches infinity. \vspace{0.0cm}
\end{proposition}
\begin{proof}
Let $p_{min}=\min_{n}q_n, p_{max}=\max_{n}q_n$. Since the DRM game is a potential game, any NEP can be reached in finite time when users update their strategy in a sequential manner (i.e., when no simultaneous updates occur) starting from any point. Thus, there exists a finite integer $U$ which is the maximal number of updates needed to reach any NEP from any starting point.

Next, consider the updating times $t_{\ell\cdot U-U+1}, t_{\ell\cdot U-U+2}, ..., t_{\ell\cdot U}$ for $\ell=1, 2, ...$. Note that given any strategy profile $\sigma^{(t_{\ell\cdot U-U})}$ by time $t_{\ell\cdot U-U}$, there exists a sequence of sequential strategy updates across users during the updating times $t_{\ell\cdot U-U+1}, t_{\ell\cdot U-U+2}, ..., t_{\ell\cdot U}$ such that the system surely reaches an equilibrium by time $t_{\ell\cdot U}$. Since the probability for each such update is greater than $p_{min}\left(1-p_{max}\right)^{N-1}$ (i.e., the desired user updates its strategy and all other $N-1$ users' strategies are remain fixed), the probability to reach an equilibrium at time $t_{\ell\cdot U}$ starting at time $t_{\ell\cdot U-U+1}$ is greater than $\left[p_{min}\left(1-p_{max}\right)^{N-1}\right]^U$ for all $\ell=1, 2, ...$\;. Similarly, the probability that the system does not reach a NEP at time $t_{\ell\cdot U}$ starting at time $t_{\ell\cdot U-U+1}$ is less than $1-\left[p_{min}\left(1-p_{max}\right)^{N-1}\right]^U$. Since this bound is independent of the starting point, the probability that the system does not reach a NEP at time $t_{\ell\cdot U}$ is less than $\left[1-\left[p_{min}\left(1-p_{max}\right)^{N-1}\right]^U\right]^{\ell}$. Thus, letting $\ell\rightarrow\infty$ completes the proof.
\end{proof}

\begin{table}
\caption{BR-DRM Algorithm}
\centering
\normalsize\begin{tabular}{|l|}
\hline
\vspace{-0.3cm}\\
\hspace{0.3cm} 1) \hspace{0.3cm} \textbf{Initialize}\\
\hspace{0.3cm}    \hspace{0.7cm} each user (say $n$) estimates $u_n(k)$ for all $k$, and
 \\ \hspace{1.1cm} selects the $M$ channels with the highest $u_n(k)$ \\\\
\hspace{0.3cm} 2) \hspace{0.4cm} \textbf{repeat} (at each updating time)\hspace{0.5cm} \\
\hspace{0.3cm} 3) \hspace{1.1cm} $\mathcal{N}_a\leftarrow$ updated set of active users \\
\hspace{0.3cm} 4) \hspace{1.1cm} \textbf{for} $n\in\mathcal{N}_a$ \textbf{do}\\
\hspace{0.3cm} 5) \hspace{1.7cm} estimate $v_n(k,\sigma_{-n})$ for all $k$ \\
\hspace{0.3cm} 6) \hspace{1.7cm} $\displaystyle k_n^*\leftarrow\left\{k_{n,1}^*, k_{n,2}^*, ..., k_{n,M}^*\right\}$ by (\ref{eq:BR_DRM_sol3})  \\
\hspace{0.3cm} 7) \hspace{1.7cm} $(k_n,p_n)\leftarrow (k_n^*,P_n)$ \\
\hspace{0.3cm} 8) \hspace{1.1cm} \textbf{end for}\\
\hspace{0.3cm} 10) \hspace{0.2cm} \textbf{until} all rates converge\\
\hline
		\end{tabular}
	\label{tab:BR_DRM}
\end{table}

\subsection{Efficiency of the BR-DRM Algorithm}
\label{ssec:efficiency}

The convergence analysis provided in Section \ref{ssec:BR_potential_game} implies that the BR-DRM algorithm converges to a stable channel allocation. However, this stable operating point may not be a system-wide optimal solution. Though simulation results demonstrate good performance of the algorithm in terms of achievable user rate, in this section we provide theoretical performance guarantee of the performance that can be expected by implementing BR-DRM. We examine the performance gain under BR-DRM (i.e., when users apply distributed learning of the dynamic load to update their strategies) as compared to a naive algorithm, in which every user chooses a channel randomly and does not apply the learning process to update its strategy. For purposes of analysis, we consider the case where the network forms a $|\emph{I}|$-regular graph, and every user experiences equal rates for all channels (when channels are free), i.e., $u_n=u_n(k)=u_n(k')$ for all $k, k'$. We set $P_n=K/\left(|\emph{I}|+1\right)$ for all $n$ (which captures proportional fairness among users as will be discussed in subsequent sections). We focus on the more interesting case where $|\emph{I}|+1>K$ (thus, best response is used to mitigate interference among neighbors) and for the ease of presentation assume that $\left(|\emph{I}|+1\right)/K \in \mathbb{Z}$. \vspace{0.0cm}
\begin{theorem}
\label{th:gain}
Assume that the assumptions presented in this section hold. Let $R_n^{BR-DRM}$, $R_n^{Naive}$ be the rate of user $n$ achieved by the BR-DRM and naive algorithms, respectively. Then, the ratio between the user rate achieved by the BR-DRM algorithm and the user rate achieved by the naive algorithm is given by:
\beq
\displaystyle\frac{R_n^{BR-DRM}}{R_n^{Naive}}\geq\eta\triangleq \frac{\left(1-\frac{K}{|\emph{I}|+1}\right)^{\frac{|\emph{I}|+1}{K}-1}}
{\left(1-\frac{1}{|\emph{I}|+1}\right)^{|\emph{I}|}}. \vspace{0.0cm}
\eeq
\end{theorem}
\begin{proof}
To prove the theorem we first lower bound the achievable expected rate of user $n$ under BR-DRM. Let $\emph{I}_n(k)$ be the set of user $n$'s neighbors who select channel $k$. Assume to the contrary that BR-DRM has converged and that user $n$ selects channel $k_1$ and $|\emph{I}_n(k_1)|>|\emph{I}|/K$. Since there exists a channel $k_2$ with $|\emph{I}_n(k_2)|\leq |\emph{I}|/K$ (and thus higher rate can be achieved over channel $k_2$), playing best response by user $n$ cannot be terminated by selecting channel $k_1$ which contradicts the assumption. Since this argument holds for every user in the system, and the BR-DRM converges (in a finite time) by Theorem \ref{th:BR_potential_game}, then $\emph{I}_n(k_n^*)\leq |\emph{I}|/K$ for all $n$ in equilibrium, where $k_n^*$ is the selected channel by user $n$ at equilibrium. As a result, the achievable rate of user $n$ is given by:
\beq
R_n^{BR-DRM}\geq\displaystyle u_n\frac{K}{|\emph{I}|+1}\cdot \left(1-\frac{K}{|\emph{I}|+1}\right)^{\frac{|\emph{I}|+1}{K}-1} \;\;\forall n.
\eeq

Next, we compute the expected user rate achieved by the naive algorithm where every user chooses a channel randomly without using CSI. Assume that user $n$ transmits over channel $k$. Note that channel $k$ is selected by all other users with a probability $1/K$ and then every user that picks channel $k$ actually transmits over it with a probability $K/\left(|\emph{I}|+1\right)$. Therefore, the expected rate of user $n$ on channel $k$ is: $R_n(k)=u_n\frac{K}{|\emph{I}|+1}\left(1-\frac{1}{K}\cdot\frac{K}{|\emph{I}|+1}\right)^{|\emph{I}|}$. Since every channel is selected with equal probability $1/K$, the expected rate of user $n$ achieved by the naive algorithm is given by:
\beq
R_n^{Naive}=\displaystyle u_n\frac{K}{|\emph{I}|+1}\left(1-\frac{1}{K}\cdot\frac{K}{|\emph{I}|+1}\right)^{|\emph{I}|} \;.
\eeq
Hence, the ratio between the expected user rate achieved by the BR-DRM algorithm and the expected user rate achieved by the naive algorithm is given by:
\beq
\displaystyle\frac{R_n^{BR-DRM}}{R_n^{Naive}}\geq\eta\triangleq
\frac{\left(1-\frac{K}{|\emph{I}|+1}\right)^{\frac{|\emph{I}|+1}{K}-1}}
{\left(1-\frac{1}{|\emph{I}|+1}\right)^{|\emph{I}|}} \;\;\;\forall n.
\eeq
\end{proof}
\begin{remark}
Note that $\lim_{\frac{|\emph{I}|+1}{K}\rightarrow 1}\left(1-\frac{K}{|\emph{I}|+1}\right)^{\frac{|\emph{I}|+1}{K}-1}=1$ and that both numerator and denominator of $\eta$ approach $e^{-1}$ as $|\emph{I}|$ increases and $K$ is fixed. Thus, it can be verified that $\eta$ is bounded by $\displaystyle 1\leq\eta\leq e$, where $\eta$ approaches $1$ as $|\emph{I}|$ approaches infinity and $K$ is fixed, and $\eta$ approaches $e$ when $|\emph{I}|+1=K$ and $K$ approaches infinity. Thus, Theorem \ref{th:gain} provides an insight about the performance gain that can be expected by the BR-DRM algorithm. Specifically, under the system model considered in this section, implementing BR-DRM guarantees that \emph{every user in the system improves its performance} (in terms of expected rate) as compared to the naive algorithm. Significant performance gain is obtained when $|\emph{I}|+1=K$ (i.e., in situations of a low collision level). In this case, we have $\eta=\left(1-1/K\right)^{1-K}$. Thus, the user rate increases by more than $100\%$ for $K\geq 2$ (since the performance gain is greater than $\eta=2$) and more than $170\%$ for very large $K$ (since the performance gain is greater than $\eta\approx e$) by implementing BR-DRM as compared to the naive algorithm.
\end{remark}

\section{Achieving Global Proportional Fairness:\\ A Cooperative Setting}
\label{sec:fairness}

Instead of solving a distributed rate maximization as done in the preceding section, here we are interested in developing a distributed algorithm that attains proportionally fair rates in the network (using information sharing between neighbors only). Cooperation in this section refers to a social behavior (by designing a social utility function for each user) that can lead to a globally-optimal operating point. Nevertheless, the model is still cast as a non-cooperative game in the sense that users act with respect to their own social utility. We consider the case where $M=1$. Thus, $k_n\in\mathcal{K}$ is a natural number denoting a single channel chosen by user $n$. Formally, the problem is to find a strategy profile that maximizes the sum-log rate in the network:
\beq
\label{eq:cooperative_optimization}
\bea{l}
\displaystyle \sigma^*=\arg\;\max_{\left\{k_n\in\mathcal{K}, 0\leq p_n\leq 1\right\}_{n=1}^{N}}\hspace{0.3cm} \sum_{n=1}^{N} \log R_n\left(\sigma\right) \;.
\ena
\eeq
The above optimization problem (\ref{eq:cooperative_optimization}) was first formulated in \cite{Kar_2004_Achieving} under a variation of the ALOHA model considered in this paper for single-channel systems (i.e., $K=1$) and equal rates for all links. In consistence with the previous section, it is convenient to view each user in the network as a player that takes actions with respect to a local utility when solving a discrete optimization problem, as suggested in \cite{Marden_2012_Revisiting}. In what follows we address this problem (\ref{eq:cooperative_optimization}) from a game theoretic perspective under the multi-channel case.

\subsection{Exact Potential Game Formulation}
\label{ssec:cooperative}

In Section \ref{sec:rate_maximization} we have shown that any NEP of the DRM game is a local maximum of its potential function (\ref{eq:BR_potential_function}). In this section, however, we are interested in finding a \emph{global maximum} of (\ref{eq:cooperative_optimization}) since it attains a global proportional fairness in the network.

Let $\mathcal{I}_n(k)$ be the set of user $n$'s neighbors that transmit over channel $k$, and let
\beq
\bea{l}
\label{eq:F}
\displaystyle F_n(k_n, p_n,\sigma_{-n})\vspace{0.0cm}\\
\displaystyle\triangleq\log\left(u_n(k_n)p_n\right)-I_n(k_n,\sigma_{-n})
-\log\left(\frac{1}{1-p_n}\right)\left|\mathcal{I}_n(k_n)\right|,
\ena
\eeq
be the \emph{cooperative utility} (or fair utility) for user $n$. Note that the cooperative utility balances between individual and social utilities. The term $\log\left(u_n(k_n)p_n\right)-I_n(k_n,\sigma_{-n})$ is the individual utility for user $n$, where $\log\left(\frac{1}{1-p_n}\right)|\mathcal{I}_n(k)|$ represents the aggregated log-interference that user $n$ causes to its neighbors. Throughout this section it is assumed that user $n$ can compute its cooperative utility when making decisions (see a discussion on a practical implementation in section \ref{ssec:NBRF}). We refer to this game as the \emph{fairness game}.

Next, we show that the fairness game is an exact potential game where $\sum_n\log R_n\left(\sigma\right)$ is a potential function of the game.
\begin{definition}[{\cite{Monderer_Potential_1996}}]
The fairness game is called an \emph{exact} potential game if there is an exact potential function $\phi : \sigma\rightarrow \mathbb{R}$ such that for every user $n$ and for every $\sigma_{-n}=\left\{k_i, p_i\right\}_{i\neq n}$, where $k_i\in\mathcal{K}$, $0\leq p_i\leq 1$, the following holds:
\beq
\label{eq:exact_potential_def}
\bea{l}
\vspace{0.0cm}\displaystyle F_n(\sigma_n^{(2)}, \sigma_{-n})-\displaystyle F_n(\sigma_n^{(1)}, \sigma_{-n}) \\ \vspace{0.0cm} \hspace{2cm}
\displaystyle=\phi(\sigma_n^{(2)}, \sigma_{-n})-\displaystyle \phi(\sigma_n^{(1)}, \sigma_{-n})
\;,\vspace{0.0cm} \\ 
\forall \sigma_n^{(1)}=(k_n^{(1)}, p_n^{(1)}), \sigma_n^{(2)}=(k_n^{(2)}, p_n^{(2)}) \;,
\vspace{0.0cm} \\ \hspace{1cm}
 k_n^{(1)}, k_n^{(2)}\in\mathcal{K} \;, 0\leq p_n^{(1)}, p_n^{(2)}\leq 1 \; .
\ena
\eeq
\end{definition}

\textsl{\theorem\label{th:exact}{
The fairness game is an exact potential game, with the following exact potential function:
\beq\label{eq:exact_potential_function}
\bea{l}
\phi(\sigma) = \displaystyle\sum_{n=1}^{N}\log R_n\left(\sigma\right)
\;.
\ena
\eeq
}}\vspace{0.0cm} \\
\begin{proof}
The proof is given in Appendix \ref{app:pr_exact}.
\end{proof}

\subsection{Nash Equilibrium of the fairness game}
\label{ssec:NEP_cooperative}

Since the fairness game is an exact potential game with an upper bounded potential function (by $\phi(\sigma)<\sum_{n=1}^{N}{\max_{k}\log\left({u}_n(k)\right)}$), any BR dynamics converges to a NEP in the sense that users cannot increase their cooperative utility by unilaterally changing their strategies. However, any local maximum of the potential function (\ref{eq:exact_potential_function}) is a NEP of the game. Thus, here we first characterize the NEPs' structure of the fairness game. In Section \ref{ssec:NBRF} we will use this result to develop an algorithm that achieves the best NEP in the sense that the global maximum of (\ref{eq:exact_potential_function}) is attained.
\vspace{0.0cm}
\begin{definition}
A Nash Equilibrium Point (NEP) for the fairness game is a strategy profile $\sigma^*=(\sigma_n^*,\sigma_{-n}^*)$, where $k_i^*\in\mathcal{K}$, $0\leq p_i^*\leq 1$ for all $i\in\mathcal{N}$, such that
\beq
\label{eq:NEP2}
\bea{l}
\displaystyle F_n\left(\sigma_n^*,\sigma_{-n}^*\right)\geq
\displaystyle F_n(\tilde{\sigma}_n,\sigma_{-n}^*)  \vspace{0.0cm}\\\hspace{1cm} \forall n \;, \forall \tilde{\sigma}_n=(\tilde{k}_n, \tilde{p}_n)\;,\; \tilde{k}_n\in\mathcal{K}\;,\;0\leq \tilde{p}_n\leq 1
  \;. \vspace{0.0cm}
\ena
\eeq
\end{definition}
\textsl{\theorem\label{th:NEP_cooperative1}{
A strategy profile $\sigma^*=\left\{k_n^*, p_n^*\right\}_{n=1}^{N}$ is a \vspace{0.0cm} NEP for the fairness game if $k_n^*\in\mathcal{K}$, $\displaystyle p_n^*=\frac{1}{\left|\mathcal{I}_n(k_n^*)\right|+1}$ for all $n\in\mathcal{N}$.
}}\vspace{0.0cm} \\
\begin{proof}
Fix a strategy profile $\sigma_{-n}$ and assume that user $n$ updates its strategy. To prove the theorem it suffices to show that for all $n$ and any $\sigma_{-n}$ the following holds:
\begin{center}
$
\bea{l}
\displaystyle\sigma_n^*=\left(k_n^*,p_n^*=\frac{1}{\left|\mathcal{I}_n(k_n^*)\right|+1}\right)
\vspace{0.0cm} \\ \hspace{0.5cm}
\displaystyle=\arg\max_{k_n\in\mathcal{K}, 0\leq p_n\leq 1} F_n(k_n, p_n, \sigma_{-n})
\vspace{0.0cm} \\ \hspace{0.5cm}
\displaystyle
=\arg\max_{k_n\in\mathcal{K}, 0\leq p_n\leq 1}\left[\log u_n(k_n)+\log p_n
\vspace{0.0cm} \right.\\\left. \hspace{2cm}
\displaystyle-I_n(k_n,\sigma_{-n})-\log\left(\frac{1}{1-p_n}\right)|\mathcal{I}_n(k_n)|
\right]\;.
\ena
$
\end{center}
Note that for any $k_n\in\mathcal{K}$ the terms $\log u_n(k_n), I_n(k_n,\sigma_{-n})$ are independent of $p_n$. Thus, it suffices to show that for any given $k_n$ the following holds:
\begin{center}
$
\bea{l}
\displaystyle \frac{1}{\left|\mathcal{I}_n(k_n)\right|+1}
\displaystyle=\arg\max_{0\leq p_n\leq 1} \log p_n
+\log\left(1-p_n\right)|\mathcal{I}_n(k_n)|
\;.
\ena
$
\end{center}
The case where $|\mathcal{I}_n(k_n)|=0$ is straightforward since user $n$ does not interfere with other users (setting $p_n=1$ maximizes the RHS by defining $0\cdot\log 0=0$). Thus, we consider the case where $|\mathcal{I}_n(k_n)|\geq 1$. Note that the function $\log p_n
+\log\left(1-p_n\right)|\mathcal{I}_n(k_n)|$ is strictly concave function of $p_n$ (for $0\leq p_n\leq 1$). Therefore, it has a unique global maximum. differentiating with respect to $p_n$ and equating to zero yields $p_n^*=\frac{1}{\left|\mathcal{I}_n(k_n)\right|+1}$ which completes the proof. \vspace{0.0cm}
\end{proof}

\begin{corollary}
\label{cor:NEP1}
A local maximum of (\ref{eq:exact_potential_function}) is attained only if every user $n$ is associated with an attempt probability $\displaystyle p_n=\frac{1}{\left|\mathcal{I}_n(k_n)\right|+1}$. In particular, the strategy profile that attains proportionally fair rates (i.e., the solution to (\ref{eq:cooperative_optimization})) must satisfy $\displaystyle p_n=\frac{1}{\left|\mathcal{I}_n(k_n)\right|+1}$ for all $n$. \vspace{0.0cm}
\end{corollary}

\begin{theorem}
\label{th:NEP_cooperative2}
Let $\left\{k_n^*\right\}_{n=1}^{N}$ be a given channel allocation for all users. A strategy profile
\beq
\sigma^*=\left\{k_n^*, p_n^*=\frac{1}{\left|\mathcal{I}_n(k_n^*)\right|+1}\right\}_{n=1}^{N}
\eeq
is the unique solution to the following optimization problem:
\beq
\label{eq:NEP_cooperative_p}
\bea{l}
\displaystyle
\left\{p_n^*\right\}_{n=1}^{N}
=\arg\max_{\left\{0\leq p_n\leq 1\right\}_{n=1}^{N}}
\sum_{n\in\mathcal{N}:k_n^*=k} \log R_n\left(\left\{k_n^*, p_n\right\}_{n=1}^{N}\right)
\vspace{0.0cm}\\\hspace{6cm}
\forall k\in\mathcal{K} \;.
\ena
\eeq
\vspace{0.0cm}
\end{theorem}
\begin{proof}
Let $\mathcal{N}_k$ be the set of users that select channel $k$. The achievable rate of user $n\in\mathcal{N}_k$ is given by:
\begin{center}
$
R_n\left(\sigma\right)=u_n(k)p_n\prod_{i\in\mathcal{I}_n(k)}\left(1-p_i\right)\;,
$
\end{center}
where $\mathcal{I}_n(k)$ is the set of user $n$'s neighbors that transmit over channel $k$. Taking log on both sides yields:
\begin{center}
$
\bea{l}
\displaystyle\log(R_n\left(\sigma\right))=\log(u_n(k))+\log(p_n)+\sum_{i\in\mathcal{I}_n(k)}
\log\left(1-p_i\right)\;, \vspace{0.0cm}\\\hspace{6cm}
 \forall n\in\mathcal{N}_k \;.
\ena
$
\end{center}
Let $L_k\triangleq\sum_{n\in\mathcal{N}_k}\log(R_n\left(\sigma\right))$ be the sum log rate on channel $k$. Hence,
\begin{center}
$
\bea{l}
\displaystyle L_k=\sum_{n\in\mathcal{N}_k}\left[\log(u_n(k))+\log(p_n) 
                \displaystyle+\sum_{i\in\mathcal{I}n(k)}
                \log\left(1-p_i\right)\right]. \vspace{0.0cm}
\ena
$
\end{center}
Note that $L_k$ is a strictly concave function of $p_n, n\in\mathcal{N}_k$. Therefore, it has a unique global maximum. Differentiating $L_k$ with respect to $p_n\;, \; n\in\mathcal{N}_k$, and equating to zero yields $p_n^*=\frac{1}{\left|\mathcal{I}_n(k)\right|+1}$ for all $n\in\mathcal{N}_k$, which completes the proof.
\vspace{0.0cm}
\end{proof}

Combining Theorems \ref{th:NEP_cooperative1} and \ref{th:NEP_cooperative2} yields: \vspace{0.0cm}

\begin{corollary}
\label{cor:NEP2}
A strategy profile $\sigma^*=\left\{k_n^*, p_n^*\right\}_{n=1}^{N}$ is a \vspace{0.0cm} NEP for the fairness game if $\left\{p_n^*\right\}_{n=1}^{N}$ solves (\ref{eq:NEP_cooperative_p}).
\vspace{0.0cm}
\end{corollary}

Corollary \ref{cor:NEP1} follows directly from the NEPs' structure characterized in Theorem \ref{th:NEP_cooperative1}. We will use the fact that attaining the global maximum of (\ref{eq:exact_potential_function}) implies $\displaystyle p_n=\frac{1}{\left|\mathcal{I}_n(k_n)\right|+1}$ for all $n$ to design a distributed learning algorithm that converges to the solution of (\ref{eq:cooperative_optimization}). Corollary \ref{cor:NEP2} sheds a light on the operating points of the system. Learning algorithms used to converge to a global optimum may spend some time at local maxima of the objective function (i.e., a NEP). Corollary \ref{cor:NEP2} shows that the local maxima of the potential function may not be so bad. Specifically, every NEP of the fairness game can be viewed as a local proportional fairness in the sense that proportionally fair rates are attained among all users that share channel $k$ for all $k\in\mathcal{K}$.

\subsection{Distributed Cooperative Learning Algorithm}
\label{ssec:NBRF}

The optimization problem in (\ref{eq:cooperative_optimization}) is a combinatorial optimization problem over a graph, and it requires a centralized solution that uses global information which is impractical in large-scale networks. Therefore, we propose a probabilistic approach to solve the problem in a distributed manner. We develop a distributed cooperative learning algorithm, dubbed Noisy BR for Fairness (NBRF) algorithm, with the goal of solving (\ref{eq:cooperative_optimization}) using limited message exchanges between neighbors only. NBRF is a cooperative algorithm in the sense that users make decisions with respect to the cooperative utility that balances between their own utilities and the interference level they cause to their neighbors.

Recall that BR dynamics may lead to local maxima of the potential function. Hence, instead of playing purely BR, in NBRF users play noisy BR (also known as spatial adaptive play or log-linear learning) when updating their strategies \cite{Blume_1993_Statistical, Young_1998_Individual, Marden_2012_Revisiting}. In NBRF, active users construct a probability mass function (pmf) over their actions and draw their actions according to this distribution. Typically, the BR is played with high probability, while other strategies are played with a probability that decays exponentially fast with the myopic utility loss in order to escape local maxima.
Specifically, the pmf over the available actions is given by:
\beq
\label{eq:pmf}
\displaystyle\Pr((k_n, p_n)=(k,p))=\frac{e^{\beta F_n(k, p, k_{-n}, p_{-n})}}{\displaystyle \sum_{k'=1}^{K}\sum_{r=1}^{|\mathcal{I}_n|+1}e^{\beta F_n(k', 1/r, k_{-n}, p_{-n})}}
\eeq
for some exploration parameter $\beta>0$. For the ease of presentation, we assume continuous random rates $u_n(k)$ to guarantee a uniqueness of the maximizer (otherwise BRs are drawn uniformly). Note that when $\beta=0$ the pmf assigns equal weights on all strategies, while the probability of playing BR approaches one as $\beta\rightarrow\infty$ (a discussion on the setting of $\beta$ based on simulated annealing analysis \cite{Hajek_1988_Cooling} is provided in the end of this section). The NBRF Algorithm is given in Table \ref{tab:NBRF}. The set of active users in Step 3 is determined according to the distributed mechanism described in Sec. \ref{sec:network}. Step $5$ requires the active users to construct the pmf given in (\ref{eq:pmf}) based on the computation of $F_n(k, p, k_{-n}, p_{-n})$ for all $k=1, ..., K$, $p=1, 1/2, ..., 1/(|\mathcal{I}_n|+1)$ given in (\ref{eq:F}). In Step 6, active users must send complete information about their updated strategies to their neighbors such that all users can compute their cooperative utility at each given updating time. A similar mechanism as described in Section \ref{sec:network} can be applied, where the pilot signal is now replaced by a packet containing complete information about the updated strategy.  Users may repeat updating strategies until their rates converge or for a predetermined number of iterations and then stick their BR (see a discussion in the end of this section).

\begin{table}
\caption{NBRF Algorithm}
\centering
\normalsize\begin{tabular}{|l|}
\hline
\vspace{-0.3cm}\\
\hspace{0.3cm} 1) \hspace{0.3cm} \textbf{Initialize}\\
\hspace{0.3cm}    \hspace{0.7cm} based on message exchanges between neighbors \\ \hspace{1.1cm}
each user (say $n$) set $k_n\leftarrow \arg \; \displaystyle\max_k \; \left\{ u_n(k) \right\}$
 \\ \hspace{1.1cm} and $p_n\leftarrow 1/\left(|\mathcal{I}_n(k_n)|+1\right)$. \\\\
\hspace{0.3cm} 2) \hspace{0.3cm} \textbf{repeat} (at each updating time)\hspace{0.5cm} \\
\hspace{0.3cm} 3) \hspace{0.9cm} $\mathcal{N}_a\leftarrow$ updated set of active users\\
\hspace{0.3cm} 4) \hspace{0.9cm} \textbf{for} $n\in\mathcal{N}_a$ \textbf{do}\\
\hspace{0.3cm} 5) \hspace{1.5cm} draw $(k_n,p_n)$ randomly according
\\
\hspace{2.35cm}
to the distribution given in (\ref{eq:pmf}). \\
\hspace{0.3cm} 6) \hspace{1.5cm} send a packet containing $(k_n,p_n)$ to \\
\hspace{2.35cm}                   inform all neighbors $\mathcal{I}_n$
\\
\hspace{0.3cm} 7) \hspace{0.9cm} \textbf{end for}\\
\hspace{0.3cm} 8) \hspace{0.2cm} \textbf{until} all rates converge \\
\hline
		\end{tabular}
	\label{tab:NBRF}
\end{table}

The following theorem shows that NBRF attains proportional fairness with an arbitrarily high probability as time increases.\vspace{0.0cm}

\begin{theorem}
\label{th:NBRF_optimal}
Let $\sigma^{NBRF(\beta)}(t), \Sigma^*$ be the strategy profile under NBRF (with a parameter $\beta$) at time $t$ and the set of strategy profiles that solves (\ref{eq:cooperative_optimization}), respectively. For any $\epsilon>0$ there exists $\beta>0$ such that
\beq
\label{eq:NBRF_convergence}
\displaystyle\lim_{t\rightarrow\infty} \Pr\left(\sigma^{NBRF(\beta)}(t)\in\Sigma^*\right)\geq 1-\epsilon \;.
\eeq
\vspace{0.0cm}
\end{theorem}
\begin{proof}
The proof is based on the results reported in Section \ref{ssec:NEP_cooperative} and the fact that a noisy best response dynamics following (\ref{eq:pmf}) in exact potential games converges to a stationary distribution of the Markov chain corresponding to the game \cite{Blume_1993_Statistical}. By Theorem \ref{th:exact}, the fairness game with the cooperative utility $F_n$ is an exact potential game with an exact potential function $\phi$ given in (\ref{eq:exact_potential_function}). Since NBRF plays noisy BR with respect to $F_n$, the stationary distribution of the strategy profile is given by \cite{Blume_1993_Statistical}:
\beq
\label{eq:stationary_dist}
\displaystyle Pr(\sigma^{NBRF(\beta)}=\sigma)=\frac{e^{\beta \phi(\sigma)}}{\displaystyle \sum_{\tilde{\sigma}}e^{\beta\phi(\tilde{\sigma})}}\;.
\eeq
Next, note that the number of user $n$'s neighbors that transmit over channel $k_n$, $|\mathcal{I}_n(k_n)|$, is lower bounded by $|\mathcal{I}_n|$ for all $n$. Therefore, following Corollary \ref{cor:NEP1}, the strategy profile $\sigma^*$ that attains the global maximum of (\ref{eq:exact_potential_function}) lies inside the action space played by NBRF. Therefore, for every $\epsilon>0$ we can choose $\beta>0$ sufficiently large such that the stationary distribution puts a sufficiently high weight on the strategy profile that maximizes (\ref{eq:exact_potential_function}) (i.e., $\phi$ in (\ref{eq:stationary_dist})). Thus, (\ref{eq:NBRF_convergence}) is satisfied as time approaches infinity.
\vspace{0.0cm}
\end{proof}

Following the proof of Theorem \ref{th:NBRF_optimal}, the stationary distribution of the homogenous Markov chain with a fixed $\beta$ corresponding to the game is given by (\ref{eq:stationary_dist}). As a result, as the probability of playing BR increases (i.e., by increasing $\beta$) the probability of attaining the global maximum of the potential function (\ref{eq:exact_potential_function}) increases with time. Achieving the optimal solution (i.e., letting $\epsilon\rightarrow 0$ in (\ref{eq:NBRF_convergence})) requires $\beta$ to approach infinity. However, increasing $\beta$ too fast may push the algorithm into a local maximum for a long time (since the probability of not playing BR is too small). Next, let $\beta=\beta(t)$ be a function of time. The process of increasing $\beta(t)$ during the algorithm is also known as \emph{cooling} the system in simulated annealing analysis, where $T(t)=1/\beta(t)$ represents the \emph{temperature}. Following simulated annealing analysis \cite{Hajek_1988_Cooling}, convergence to the optimal solution is attained by increasing $\beta(t)$ as $\beta(t)=\log(t)/\Delta, t=1, 2, ...\;$ ($\Delta$ is a constant and will be discussed in the sequel). As a result, users explore strategy profiles in the beginning of the algorithm and will stick their BR as time approaches infinity. In cases where the optimal operating point is not unique, the algorithm may converge to one of the optimal operating points. An alternative way is to set a piecewise constant $\beta(t)$ over time as suggested in \cite{Tsitsiklis_1988_Survey}. Let $\left\{t_k^{\beta}\right\}$ be an increasing time sequence, with $t_1^{\beta}=1$, where $\beta(t)=k$ is kept fixed for all $t_k^{\beta}\leq t\leq t_{k+1}^{\beta}$. Intuitively  speaking, the total time between $t_k^{\beta}$ and $t_{k+1}^{\beta}$ should be large enough, such that the stationary Markov chain associated with the system under a fixed $\beta(t)=k$ will approach arbitrary close to a steady state (with a stationary distribution given in (\ref{eq:stationary_dist})). Following simulated annealing analysis in \cite{Tsitsiklis_1988_Survey}, it suffices to let $t_{k+1}^{\beta}-t_k^{\beta}=e^{k\Delta}$. Note that the piecewise constant update has a logarithmic order with time since $\beta(t_k^{\beta})=k\approx\log t_k^{\beta}/\Delta$. It suffices to set the constant $\Delta$ to be greater than the maximal change in the objective function. Since the maximal value of the objective function is upper bounded by $N\log \max_{n, k_n} u_n(k_n)$ and the minimal value is lower bounded by $N(\log\left(\frac{\min_{n, k_n} u_n(k_n)}{\max_n |I_n|+1}\right)-\max_n |I_n|\log 2)$, it suffices to set $\Delta>N\left(\log\left(\max_{n, k_n} u_n(k_n)\right)-\log\left(\frac{\min_{n, k_n} u_n(k_n)}{\max_n |I_n|+1}\right)\right.$ $\left.+\max_n |I_n|\log 2\right)$ to achieve convergence. It should be noted, however, that simulation results demonstrate fast convergence to the optimal solution with much smaller values of $\Delta$ under typical scenarios.

\section{Numerical Examples}
\label{sec:simulation}

In this section we provide numerical examples to illustrate the performance of the algorithms. We simulated the following network: $N$ users were randomly dispersed (uniformly) in a circle region with a radius of $10$ meters. Each user causes interference to all users in a radius of $5$ meters. Every user can choose one channel for transmission among $K$ channels. We assume equal achievable rates $u_n(k)=100$Mbps for all users on all channels when channels are free (i.e., collision-free utility). We performed $1,000$ Monte-Carlo experiments and averaged the performance over experiments. The randomness for each trial over which the average performance is plotted comes from the random dynamic nature of the user updates (thus, each experiment results in a different update dynamic and might even converge to a different equilibrium point).

We first consider the distributed rate maximization problem under the non-cooperative setting, where each user maximizes its own rate under a constraint on the attempt probability. The estimation of $v_n$ is based on a moving window of $100$ packets. We first examine a small connected network with $N=10$ users sharing $K=2$ channels, so as the centralized optimal exhaustive search solution (in terms of sum-log rate) can be computed and serve as a benchmark for comparison. An illustration of the small network is depicted in Fig. \ref{fig:small_network}. In Fig. \ref{fig:non_coop_N_10_K_2} we present the average rate to demonstrate the performance in terms of efficiency under fixed attempt probabilities $P=2/3$ for all users. Though optimality is not guaranteed under BR-DRM due its greedy nature, it can be seen that under this small network model BR-DRM converges to the optimal channel allocation (in terms of sum-log rate) in finite time and significantly improves performance as compared to a random channel allocation. Next, we examine the case of a large network, in which the number of users varies during time to demonstrate the robustness of the proposed algorithm. We initialized the network size by $N=250$ users, where $N/2$ users are allowed to transmit with attempt probability $P=0.7$ (e.g., primary or high-priority users) and $N/2$ users are allowed to transmit with attempt probability $P=0.3$ (e.g., secondary or low-priority users). We set the number of channels to $K=30$ (i.e., a channel represents a subsets of subcarriers as in OFDMA or allocation to PALs in the context of spectrum sharing) (in this case computing the optimal solution is intractable). We first increase the network size by adding $10$ users after $100$ iterations. Then, we increase the network size more aggressively by adding another $40$ users. It can be seen that BR-DRM converges to the equilibrium points very fast, and significantly outperforms the naive algorithm for all time instants. Note that we can further increase the robustness of the algorithms by allowing users to update their strategies only when they improve their rates by more than a predefined value.
\begin{figure}[htbp]
\centering \epsfig{file=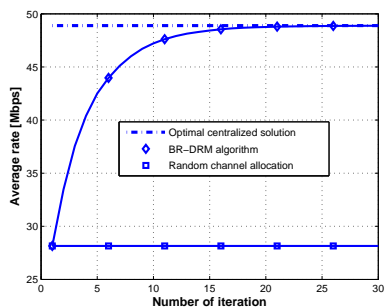,
width=0.35\textwidth}
\caption{Average rate as a function of the number of iterations. A wireless network containing 10 users and 2 channels. Each user transmits with an attempt probability $2/3$.}
\label{fig:non_coop_N_10_K_2}
\end{figure}

\begin{figure}[htbp]
\centering \epsfig{file=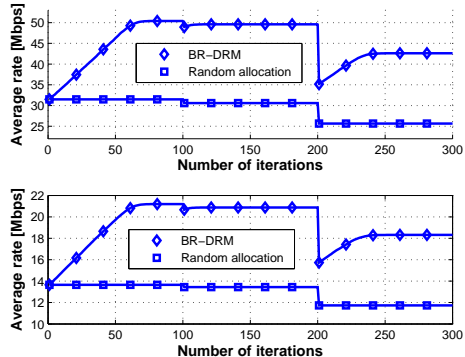,
width=0.42\textwidth}
\caption{Average rate as a function of the number of iterations under a time varying network size. $N=250, 260, 300$ for $1\leq t<100, 100\leq t<200, 200\leq t\leq 300$, respectively (where $t$ denotes the iteration index). In the top figure, the average rate of $N/2$ users with attempt probability $0.7$ is presented. In the bottom figure, the average rate of $N/2$ users with attempt probability $0.3$ is presented.}
\label{fig:non_coop_N_300_K_30}
\end{figure}

Second, we consider the cooperative setting, where the goal is to find a channel allocation and attempt probabilities in a distributed manner so as to attain proportionally fair rates among users. We compare the NBRF algorithm, given in Table \ref{tab:NBRF}, with the random channel allocation scheme, where the optimal attempt probabilities were set under any random channel allocation (i.e., $p_n=1/(|I_n(k_n)|+1)$ for all $n$). In the NBRF algorithm, we set $\beta=\log t$ (where $t=1, 2, ...$ indicates the iteration number) to construct the pmf in Step 6. We first examine a small connected network with $N=10$ users sharing $K=2$ channels, so as the centralized optimal exhaustive search solution can be computed and serve as a benchmark for comparison. An illustration of the small network is depicted in Fig. \ref{fig:small_network}. In Fig. \ref{fig:coop_N_10_K_2} we present the average log rate to demonstrate the performance in terms of proportional fairness and also the average rate to demonstrate the achievable effective rates. It can be seen that NBRF significantly improves performance as compared to a random channel allocation (even though the attempt probabilities are optimal given any random channel allocation) in terms of both fairness and efficiency. It can be seen that NBRF approaches the optimal centralized solution as time increases. This result demonstrates the efficiency of the proposed distributed learning algorithm in achieving the global proportional fairness in the network.

Next, we consider a large network, in which the number of users varies during time to demonstrate the robustness of the proposed NBRF algorithm. We initialized the network size by $N=80$ users, and set the number of channels to $K=10$ (in this case computing the optimal solution is intractable). We first increase the network size by adding $5$ users after $200$ iterations. Then, we increase the network size more aggressively by adding another $15$ users. It can be seen that NBRF approaches the equilibrium points very fast and significantly outperforms the random channel allocation for all time instants.

\begin{figure}[htbp]
\centering \epsfig{file=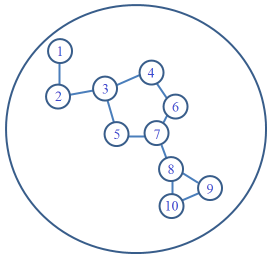,
width=0.28\textwidth}
\caption{An illustration of a small connected network with $10$ users spatially distributed in a circle area of radius $10$ meters. The users share $2$ channels. Each pair of users with distance less than $2$ meters (represented by an edge) cause mutual interference when transmitting simultaneously over the same channel.}
\label{fig:small_network}
\end{figure}

\begin{figure}[htbp]
\centering \epsfig{file=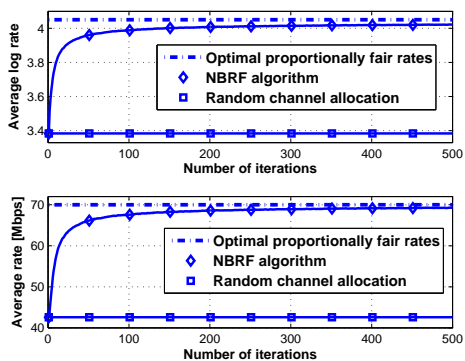,
width=0.42\textwidth}
\caption{Average sum-log rate and average rate as a function of the number of iterations. A wireless network containing 10 users and 2 channels.}
\label{fig:coop_N_10_K_2}
\end{figure}

\begin{figure}[htbp]
\centering \epsfig{file=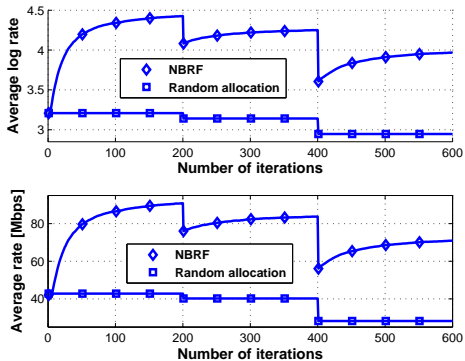,
width=0.42\textwidth}
\caption{Average sum-log rate and average rate as a function of the number of iterations under a time varying network size. $N=80, 85, 100$ for $1\leq t<200, 200\leq t<400, 400\leq t\leq 600$, respectively (where $t$ denotes the iteration index). A wireless network containing $80-100$ users and $10$ channels.}
\label{fig:coop_N_50_K_5}
\end{figure}

\section{Conclusion}
\label{sec:conclusion}

The distributed optimization problem over multiple collision channels shared by spatially distributed users was considered. We examined both the non-cooperative and cooperative settings. Under the non-cooperative setting, we developed a distributed learning algorithm for the distributed rate maximization problem, in which each user maximizes its own rate irrespective of other user utilities. Convergence was proved using the theory of best-response potential games. Under the cooperative setting, we developed a distributed cooperative learning algorithm to achieve the global proportional fairness in the networks. While direct computation of the optimal solution is impractical in large-scale networks, we showed that the proposed distributed algorithm converges to the global optimum with high probability as time increases. Simulation results demonstrated strong performance of the algorithms.

Future research directions are convergence time analysis of the proposed algorithms, analyzing their performance under malicious/malfunctioning nodes, and extensions of the NEPs efficiency analysis under the non-cooperative setting.

\section{Appendix}
\label{app}

\subsection{Occurrence of Cycles in the DRM Game Under Better-Response Dynamics}
\label{app:example}

In Theorem \ref{th:BR_potential_game} in Section \ref{ssec:cooperative} we have shown that the DRM game is a best-response potential game for $M\geq 1$ (i.e., no cycles occur when a best-response dynamics is implemented). Here, we provide an example that shows that cycles may occur when a better-response dynamics is implemented for $M>1$. Assume $N=2$ users, $K=4$ channels, $M=2$ and $P_1=P_2=0.5$. Consider the following utility matrix:
\beq
\label{eq:example_utility_matrix}
\bea{l}
\mathbf{U} \triangleq
\left[ \begin{matrix}
u_1(1) & u_1(2) & u_1(3) & u_1(4) \\
u_2(1) & u_2(2) & u_2(3) & u_2(4) \\
\end{matrix} \right]
=
\left[ \begin{matrix}
1 & 2 & 1 & 2 \\
2 & 1 & 2 & 1 \\
\end{matrix} \right]
\ena
\eeq
and an initial strategy profile:
\beq
\label{eq:example_sigma0}
\bea{l}
\sigma^{(0)} \triangleq
\left[ \begin{matrix}
P_1 & P_1 & 0 & 0 \\
0 & P_2 & P_2 & 0 \\
\end{matrix} \right]
=
\left[ \begin{matrix}
0.5 & 0.5 & 0 & 0 \\
0 & 0.5 & 0.5 & 0 \\
\end{matrix} \right]\;.
\ena
\eeq
Next, we present a better-response dynamics which results in a cycle. Assume updating time $t_1$ and let user $1$ update its strategy by switching from channels $1, 2$ (with rate $0.5\cdot 1+0.25\cdot 2=1$) to channels $3, 4$ (with a higher rate $0.5\cdot 2+0.25\cdot 1=1.25$):
\beq
\label{eq:example_sigma1}
\bea{l}
\sigma^{(1)} =
\left[ \begin{matrix}
0 & 0 & 0.5 & 0.5 \\
0 & 0.5 & 0.5 & 0 \\
\end{matrix} \right]\;.
\ena
\eeq
At updating time $t_2$ user $2$ updates its strategy by switching from channels $2, 3$ (with rate $0.5\cdot 1+0.25\cdot 2=1$) to channels $1, 4$ (with a higher rate $0.5\cdot 2+0.25\cdot 1=1.25$):
\beq
\label{eq:example_sigma2}
\bea{l}
\sigma^{(2)} =
\left[ \begin{matrix}
0 & 0 & 0.5 & 0.5 \\
0.5 & 0 & 0 & 0.5 \\
\end{matrix} \right]\;.
\ena
\eeq
At updating time $t_3$ user $1$ updates its strategy by switching from channels $3, 4$ (with rate $0.5\cdot 1+0.25\cdot 2=1$) to channels $1, 2$ (with a higher rate $0.5\cdot 2+0.25\cdot 1=1.25$):
\beq
\label{eq:example_sigma3}
\bea{l}
\sigma^{(3)} =
\left[ \begin{matrix}
0.5 & 0.5 & 0 & 0 \\
0.5 & 0 & 0 & 0.5 \\
\end{matrix} \right]
\ena
\eeq
At updating time $t_4$ user $2$ updates its strategy by switching from channels $1, 4$ (with rate $0.5\cdot 1+0.25\cdot 2=1$) to channels $2, 3$ (with a higher rate $0.5\cdot 2+0.25\cdot 1=1.25$):
\beq
\label{eq:example_sigma4}
\bea{l}
\sigma^{(4)} =
\left[ \begin{matrix}
0.5 & 0.5 & 0 & 0 \\
0 & 0.5 & 0.5 & 0 \\
\end{matrix} \right]
=\sigma^{(0)}\;.
\ena
\eeq

\subsection{Proof of Theorem \ref{th:BR_potential_game}}
\label{app:pr_ordinal}
Consider two strategies for user $n_0$, $\sigma_{n_0}^{(1)}=(k_{n_0}^{(1)}, P_{n_0})$, $\sigma_{n_0}^{(2)}=(k_{n_0}^{(2)}, P_{n_0})$, and fix the strategy profile for all other users $\sigma_{-n_0}$. Throughout the proof, the superscript $(i)$ refers to the user strategies given that user $n_0$ plays strategy $\sigma_{n_0}^{(i)}$, for $i=1,2$. The term $I_n(k, k', P, \sigma_{-n,n_0})$ refers to the log-interference function (\ref{eq:log_interference}) when user $n$ chooses channel $k$, user $n_0$ plays strategy $(k', P)$, and all other users except users $n, n_0$ play strategy profile $\sigma_{-n,n_0}$. \vspace{0.0cm} \\
\emph{Step 1: The Improvement in User $n_0$'s Rate:} \vspace{0.0cm}\\
Assume that $\sigma_{n_0}^{(2)}$ is a BR strategy for user $n_0$, i.e., \beq
R_{n_0}(\sigma_{n_0}^{(2)}, \sigma_{-{n_0}})-R_{n_0}(\sigma_{n_0}^{(1)}, \sigma_{-{n_0}})\geq 0
\eeq
for all $\sigma_{n_0}^{(1)}$, such that $k_{n_0}^{(1)}\in\mathcal{K}_M$. %
Let
$\left\{k_{n_0,1}^{(2)}, k_{n_0,2}^{(2)}, ..., k_{n_0,K}^{(2)}\right\}$
be a permutation of $\left\{1, ..., K\right\}$ such that:
\beq
\bea{l}
u_{n_0}(k_{n_0,1}^{(2)})v_{n_0}(k_{n_0,1}^{(2)},\sigma_{-n_0})\geq u_{n_0}(k_{n_0,2}^{(2)})v_{n_0}(k_{n_0,2}^{(2)},\sigma_{-n_0}) \vspace{0.0cm} \\ \hspace{3cm}
\geq\cdots\geq u_{n_0}(k_{n_0,K}^{(2)})v_{n_0}(k_{n_0,K}^{(2)},\sigma_{-n_0})\;.
\ena
\eeq
Following (\ref{eq:non_cooperative_optimization}), the BR channel-selection $k_{n_0}^{(2)}$ is given by:
\beq
k_{n_0}^{(2)}=\left\{k_{n_0,1}^{(2)}, k_{n_0,2}^{(2)}, ..., k_{n_0,M}^{(2)}\right\}\;.
\eeq
Next, arrange the entries of
$\sigma_{n_0}^{(1)}=\left\{k_{n_0,1}^{(1)}, k_{n_0,2}^{(1)}, ..., k_{n_0,K}^{(1)}\right\}$
such that:
\beq
\bea{l}
u_{n_0}(k_{n_0,1}^{(1)})v_{n_0}(k_{n_0,1}^{(1)},\sigma_{-n_0})\geq u_{n_0}(k_{n_0,2}^{(1)})v_{n_0}(k_{n_0,2}^{(1)},\sigma_{-n_0}) \vspace{0.0cm} \\ \hspace{3cm}
\geq\cdots\geq u_{n_0}(k_{n_0,M}^{(1)})v_{n_0}(k_{n_0,M}^{(1)},\sigma_{-n_0})\;.
\ena
\eeq
As a result, by the construction we have:
\beq
\bea{l}
u_{n_0}(k_{n_0,i}^{(2)})v_{n_0}(k_{n_0,i}^{(2)},\sigma_{-n_0})\geq u_{n_0}(k_{n_0,i}^{(1)})v_{n_0}(k_{n_0,i}^{(1)},\sigma_{-n_0}) \vspace{0.0cm} \\
\hspace{4cm}
\forall i=1, ..., M\;.
\ena
\eeq
Next, define
\beq
\bea{l}
\tilde{k}_{n_0}^{(1)}\triangleq k_{n_0}^{(1)}\setminus k_{n_0}^{(2)}
=\left\{\tilde{k}_{n_0,1}^{(1)}, ..., \tilde{k}_{n_0,L}^{(1)}\right\}
\;, \vspace{0.0cm} \\
\tilde{k}_{n_0}^{(2)}\triangleq k_{n_0}^{(2)}\setminus k_{n_0}^{(1)}
=\left\{\tilde{k}_{n_0,1}^{(2)}, ..., \tilde{k}_{n_0,L}^{(2)}\right\}
 \;.
\ena
\eeq
For example, if user $n_0$ selects channels $k_{n_0}^{(1)}=\left\{1, 2, 3\right\}$ and $k_{n_0}^{(2)}=\left\{3, 4, 5\right\}$ according to strategies $\sigma_{n_0}^{(1)}$ and $\sigma_{n_0}^{(2)}$, respectively, then $\tilde{k}_{n_0}^{(1)}=\left\{1, 2\right\}$, $\tilde{k}_{n_0}^{(2)}=\left\{4, 5\right\}$, and $L=2$. Note that $\tilde{k}_{n_0}^{(1)}, \tilde{k}_{n_0}^{(2)}$ have the same cardinality (say $L\triangleq|\tilde{k}_{n_0}^{(1)}|=|\tilde{k}_{n_0}^{(2)}|\leq M$) and denote the differences in the chosen channels under strategies $\sigma_{n_0}^{(1)}, \sigma_{n_0}^{(2)}$ (i.e., $\tilde{k}_{n_0}^{(1)}\cap\tilde{k}_{n_0}^{(2)}=\emptyset$). We arrange
$\left\{\tilde{k}_{n_0,1}^{(i)}, ..., \tilde{k}_{n_0,L}^{(i)}\right\}$ such that:
\beq
\bea{l}
u_{n_0}(\tilde{k}_{n_0,1}^{(i)})v_{n_0}(\tilde{k}_{n_0,1}^{(i)},\sigma_{-n_0})\geq u_{n_0}(\tilde{k}_{n_0,2}^{(i)})v_{n_0}(\tilde{k}_{n_0,2}^{(i)},\sigma_{-n_0}) \vspace{0.0cm} \\ \hspace{3cm}
\geq\cdots\geq u_{n_0}(\tilde{k}_{n_0,L}^{(i)})v_{n_0}(\tilde{k}_{n_0,L}^{(i)},\sigma_{-n_0})\;,
\ena
\eeq
for $i=1, 2$. \\
By the construction and using the monotonicity of the logarithm, we obtain
\beq
\label{eq:app_BR_potential1}
\bea{l}
\displaystyle\Delta\tilde{R}_{n_0,i}\left(\tilde{k}_{n_0,i}^{(1)}, \tilde{k}_{n_0,i}^{(2)},\sigma_{-n_0}\right)\triangleq
\vspace{0.0cm} \\ \hspace{0.5cm}
\displaystyle\log(u_n(\tilde{k}_{n_0,i}^{(2)}))-I_n(\tilde{k}_{n_0,i}^{(2)},\sigma_{-n_0})
\vspace{0.0cm} \\
\displaystyle
-\left(\log(u_n(\tilde{k}_{n_0,i}^{(1)}))-I_n(\tilde{k}_{n_0,i}^{(1)},\sigma_{-n_0})\right)
\geq 0
\;\forall i=1, ..., L\;. \vspace{0.0cm}
\ena
\eeq
\emph{Step 2: The difference in the Potential Function:} \vspace{0.0cm}\\
Next, to prove the theorem we need to show that
\beq
\phi(\sigma_{n_0}^{(2)}, \sigma_{-{n_0}})-\phi(\sigma_{n_0}^{(1)}, \sigma_{-{n_0}})\geq 0
\eeq
The difference in the proposed function (\ref{eq:BR_potential_function}) $\Delta\phi$ is given by:
\begin{center}
$\bea{l}
\vspace{0.0cm}
\Delta\phi\left(\sigma_{n_0}^{(1)}, \sigma_{n_0}^{(2)},\sigma_{-n_0}\right)
\triangleq\displaystyle \phi\left(\sigma_{n_0}^{(2)},\sigma_{-n_0}\right)-\displaystyle \phi\left(\sigma_{n_0}^{(1)},\sigma_{-n_0}\right) \vspace{0.0cm}\\ \hspace{0.0cm}
=
\displaystyle\sum_{n=1}^{N}\log\left(\frac{1}{1-P_n}\right)
\times\vspace{0.0cm}\\ \hspace{2cm}
\displaystyle\sum_{i=1}^{M}\left(\log u_{n}(k_{n,i}^{(2)})-\frac{I_n(k_{n,i}^{(2)},\sigma_{-n}^{(2)})}{2} \right) \vspace{0.0cm}\\ \hspace{0.3cm}
-\displaystyle\sum_{n=1}^{N}\log\left(\frac{1}{1-P_n}\right)
\times\vspace{0.0cm}\\ \hspace{2cm}
\displaystyle\sum_{i=1}^{M}\left(\log u_{n}(k_{n,i}^{(1)})-\frac{I_n(k_{n,i}^{(1)},\sigma_{-n}^{(1)})}{2} \right) \\ \hspace{1cm}
\ena$
\end{center}
\vspace{-1.0cm}
\begin{center}
$\bea{l}
\vspace{0.0cm}
\overset{(a)}{=}\displaystyle\sum_{i=1}^{L}\left[
\sum_{n\in\mathcal{I}_{n_0}:\tilde{k}_{n_0,i}^{(1)}\in k_n}\log\left(\frac{1}{1-P_n}\right)
\times\vspace{0.0cm}
\right. \\ \hspace{1cm}
\displaystyle\left(\log u_{n}(\tilde{k}_{n_0,i}^{(1)})-\frac{I_n(\tilde{k}_{n_0,i}^{(1)},\tilde{k}_{n_0,i}^{(2)},P_{n_0},\sigma_{-n,n_0})}{2} \right) \vspace{0.0cm}\\ \hspace{0.3cm}
\displaystyle+\sum_{n\in\mathcal{I}_{n_0}:\tilde{k}_{n_0,i}^{(2)}\in k_n}\log\left(\frac{1}{1-P_n}\right)
\times\vspace{0.0cm}\\ \hspace{1cm}
\displaystyle\left(\log u_{n}(\tilde{k}_{n_0,i}^{(2)})-\frac{I_n(\tilde{k}_{n_0,i}^{(2)},\tilde{k}_{n_0,i}^{(2)},P_{n_0},\sigma_{-n,n_0})}{2} \right) \vspace{0.0cm}\\ 
\displaystyle+\log\left(\frac{1}{1-P_{n_0}}\right)
\displaystyle\left(\log u_{n_0}(\tilde{k}_{n_0,i}^{(2)})-\frac{I_n(\tilde{k}_{n_0,i}^{(2)},\sigma_{-n_0})}{2} \right) \vspace{0.0cm}\\ \hspace{0.3cm}
\ena$
\end{center}
\vspace{-1.0cm}
\begin{center}
$\bea{l}
\vspace{0.0cm}
\displaystyle-\sum_{n\in\mathcal{I}_{n_0}:\tilde{k}_{n_0,i}^{(1)}\in k_n}\log\left(\frac{1}{1-P_n}\right)
\times\vspace{0.0cm}\\ \hspace{1cm}
\displaystyle\left(\log u_{n}(\tilde{k}_{n_0,i}^{(1)})-\frac{I_n(\tilde{k}_{n_0,i}^{(1)},\tilde{k}_{n_0,i}^{(1)},P_{n_0},\sigma_{-n,n_0})}{2} \right) \vspace{0.0cm}\\ \hspace{0.3cm}
\displaystyle-\sum_{n\in\mathcal{I}_{n_0}:\tilde{k}_{n_0,i}^{(2)}\in k_n}\log\left(\frac{1}{1-P_n}\right)
\times\vspace{0.0cm}\\ \hspace{1cm}
\displaystyle\left(\log u_{n}(\tilde{k}_{n_0,i}^{(2)})-\frac{I_n(\tilde{k}_{n_0,i}^{(2)},\tilde{k}_{n_0,i}^{(1)},P_{n_0},\sigma_{-n,n_0})}{2} \right) \vspace{0.0cm}\\ \hspace{0.3cm}
\displaystyle-\log\left(\frac{1}{1-P_{n_0}}\right)
\times\vspace{0.0cm}\\ \left. \hspace{1cm}
\displaystyle\left(\log u_{n_0}(\tilde{k}_{n_0,i}^{(1)})-\frac{I_n(\tilde{k}_{n_0,i}^{(1)},\sigma_{-n_0})}{2} \right) \right]
\displaystyle\triangleq\sum_{i=1}^L f(i) \;.
\ena$
\end{center}
Equality (a) follows since only users in $\mathcal{I}_{n_0}$ that transmit over channels $\tilde{k}_{n_0,i}^{(1)}, \tilde{k}_{n_0,i}^{(2)}$ experience a change in their interference level. Thus, it suffices to show that every term in the summation is positive (i.e., $f(i)\geq 0$ for all $i=1, ..., L$). After rearranging terms we have:
\begin{center}
$\bea{l}
\displaystyle f(i)=
\displaystyle-\sum_{n\in\mathcal{I}_{n_0}:\tilde{k}_{n_0,i}^{(1)}\in k_n}\log\left(\frac{1}{1-P_n}\right)
\times
\vspace{0.0cm}\\ \hspace{1.5cm}
\displaystyle \frac{I_n(\tilde{k}_{n_0,i}^{(1)},\tilde{k}_{n_0,i}^{(1)}, P_{n_0}, \sigma_{-n,n_0})+\log(1-P_{n_0})}{2} \vspace{0.0cm}\\ \hspace{0.3cm}
\displaystyle-\sum_{n\in\mathcal{I}_{n_0}:\tilde{k}_{n_0,i}^{(2)}\in k_n}{{\log\left(\frac{1}{1-P_n}\right)}{ \displaystyle \frac{I_n(\tilde{k}_{n_0,i}^{(2)},\tilde{k}_{n_0,i}^{(2)}, P_{n_0}, \sigma_{-n,n_0})}{2}}} \vspace{0.0cm}\\ \hspace{0.3cm}
\ena$
\end{center}
\vspace{-1.0cm}
\begin{center}
$\bea{l}
\displaystyle+\log\left(\frac{1}{1-P_{n_0}}\right)\left( \log u_{n_0}(\tilde{k}_{n_0,i}^{(2)})-\frac{I_{n_0}(\tilde{k}_{n_0,i}^{(2)},\sigma_{-n_0})}{2} \right) \vspace{0.0cm}\\ \hspace{0.3cm}
\displaystyle+\sum_{n\in\mathcal{I}_{n_0}:\tilde{k}_{n_0,i}^{(1)}\in k_n}{{\log\left(\frac{1}{1-P_n}\right)}{ \displaystyle \frac{I_n(\tilde{k}_{n_0,i}^{(1)},\tilde{k}_{n_0,i}^{(1)}, P_{n_0}, \sigma_{-n,n_0})}{2}}} \vspace{0.0cm}\\ \hspace{0.3cm}
\displaystyle+\sum_{n\in\mathcal{I}_{n_0}:\tilde{k}_{n_0,i}^{(2)}\in k_n}\log\left(\frac{1}{1-P_n}\right)
\times
\vspace{0.0cm}\\ \hspace{1.5cm}
 \displaystyle \frac{I_n(\tilde{k}_{n_0,i}^{(2)},\tilde{k}_{n_0,i}^{(2)}, P_{n_0}, \sigma_{-n,n_0})+\log(1-P_{n_0})}{2} \vspace{0.0cm}\\ \hspace{0.3cm}
\displaystyle-\log\left(\frac{1}{1-P_{n_0}}\right)\left( \log u_{n_0}(\tilde{k}_{n_0,i}^{(1)})-\frac{I_{n_0}(\tilde{k}_{n_0,i}^{(1)},\sigma_{-n_0})}{2} \right) \;,
\ena$
\end{center}
where the last equality follows by the fact the user $n_0$ contributes $-\log(1-P_{n_0})$ to the log-interference when transmitting over a channel. Hence, after rearranging terms we have:
\begin{center}
$\bea{l}
\displaystyle f(i)=
\displaystyle\log\left(\frac{1}{1-P_{n_0}}\right)\frac{I_{n_0}(\tilde{k}_{n_0,i}^{(1)},\sigma_{-n_0})}{2} \vspace{0.0cm}\\ \hspace{0.3cm}
\displaystyle-\log\left(\frac{1}{1-P_{n_0}}\right)\frac{I_{n_0}(\tilde{k}_{n_0,i}^{(2)},\sigma_{-n_0})}{2} \vspace{0.0cm}\\ \hspace{0.3cm}
\displaystyle+\log\left(\frac{1}{1-P_{n_0}}\right)\left( \log u_{n_0}(\tilde{k}_{n_0,i}^{(2)})-\frac{I_{n_0}(k^{(2)},\sigma_{-n_0})}{2} \right) \vspace{0.0cm}\\ \hspace{0.3cm}
\displaystyle-\log\left(\frac{1}{1-P_{n_0}}\right)\left( \log u_{n_0}(\tilde{k}_{n_0,i}^{(1)})-\frac{I_{n_0}(k^{(1)},\sigma_{-n_0})}{2} \right) \\ \hspace{1cm}
\ena$
\end{center}
\vspace{-1.0cm}
\begin{center}
$\bea{l}
\displaystyle
=\log\left(\frac{1}{1-P_{n_0}}\right)\Delta\tilde{R}_{n_0,i}\left(\tilde{k}_{n_0,i}^{(1)}, \tilde{k}_{n_0,i}^{(2)},\sigma_{-n_0}\right)\geq 0
\ena$
\end{center}
for all $i$. Hence, (\ref{eq:ordinal_potential_def}) follows.
Furthermore, $\phi(\sigma)$ is upper bounded by $\phi(\sigma)\leq M\sum_{n=1}^{N}{\log\left(\frac{1}{1-P_n}\right)\max_{k}\log\left({u}_n(k)\right)}$. %
As a result, $\phi(\sigma)$ in (\ref{eq:BR_potential_function}) is a bounded best-response potential function of the DRM game which completes the proof.
\newcommand*{\QEDA}{\hfill\ensuremath{\blacksquare}}%
\QEDA

\subsection{Proof of Theorem \ref{th:exact}}
\label{app:pr_exact}

Consider two strategies for user $n_0$, $\sigma_{n_0}^{(1)}=(k_{n_0}=k^{(1)}, p_{n_0}=p^{(1)})$, $\sigma_{n_0}^{(2)}=(k_{n_0}=k^{(2)}, p_{n_0}=p^{(2)})$, and fix the strategy profile for all other users $\sigma_{-n_0}$. Throughout the proof, the superscript $(i)$, refers to the user strategies given that user $n_0$ plays strategy $\sigma_{n_0}^{(i)}$, for $i=1,2$. The term $I_n(k_n, k^{(i)}, p^{(i)}, \sigma_{-n,n_0})$ refers to the log-interference function (\ref{eq:log_interference}) when user $n$ chooses channel $k_n$, user $n_0$ plays strategy $(k^{(i)}, p^{(i)})$, and all other users except users $n, n_0$ play strategy $\sigma_{-n,n_0}$. The difference in the payoff function $\Delta\tilde{R}_{n_0}$ is given by:
\begin{center}
$
\bea{l}
\displaystyle F_{n_0}(\sigma_n^{(2)}, \sigma_{-n})-F_{n_0}(\sigma_n^{(1)}, \sigma_{-n}) \vspace{0.0cm}\\
=
\displaystyle\left[\log \left(u_{n_0}(k^{(2)})p^{(2)}\right)-I_{n_0}\left(k^{(2)},\sigma_{-n_0}\right)
\right.\vspace{0.0cm}\\\left. \hspace{3cm}
\displaystyle-\log\left(\frac{1}{1-p^{(2)}}\right)\left|\mathcal{I}_{n_0}(k^{(2)})\right|\right]
\vspace{0.0cm}\\\hspace{0.3cm}
-
\displaystyle\left[\log \left(u_{n_0}(k^{(1)})p^{(1)}\right)-I_{n_0}\left(k^{(1)},\sigma_{-n_0}\right)
\right.\vspace{0.0cm}\\\left. \hspace{3cm}
\displaystyle-\log\left(\frac{1}{1-p^{(1)}}\right)\left|\mathcal{I}_{n_0}(k^{(1)})\right|\right]
\vspace{0.0cm}\\\hspace{2cm}
\triangleq \Delta F_{n_0}\left(\sigma^{(1)}, \sigma^{(2)},\sigma_{-n_0}\right) \;.
\ena
$
\end{center}

We prove the theorem for $k^{(1)}\neq k^{(2)}$. The case where $k^{(1)}=k^{(2)}$ follows similarly with minor modifications. The difference in the proposed function (\ref{eq:exact_potential_function}) $\Delta\phi$ is given by:
\begin{center}
$\bea{l}
\vspace{0.0cm}
\Delta\phi\left(\sigma_{n_0}^{(1)}, \sigma_{n_0}^{(2)},\sigma_{-n_0}\right)
\triangleq\displaystyle \phi\left(\sigma_{n_0}^{(2)},\sigma_{-n_0}\right)-\displaystyle \phi\left(\sigma_{n_0}^{(1)},\sigma_{-n_0}\right) \vspace{0.0cm}\\ \hspace{0.0cm}
=
\displaystyle\sum_{n\neq n_0}\left[\log u_{n}(k_n)+\log p_n
 \displaystyle -I_n(k_n, k^{(2)}, p^{(2)}, \sigma_{-n,n_0})\right] \vspace{0.0cm}\\ \hspace{1cm}
\displaystyle+\log u_{n_0}(k^{(2)})+\log p^{(2)}-I_{n_0}(k^{(2)},\sigma_{-n_0}) \vspace{0.0cm}\\ \hspace{0.3cm}
\displaystyle-\sum_{n\neq n_0}\left[\log u_{n}(k_n)+\log p_n
 \displaystyle -I_n(k_n, k^{(1)}, p^{(1)}, \sigma_{-n,n_0})\right] \vspace{0.0cm}\\ \hspace{1cm}
\displaystyle+\log u_{n_0}(k^{(1)})+\log p^{(1)}-I_{n_0}(k^{(1)},\sigma_{-n_0}) \hspace{1cm}
\ena$
\end{center}
\begin{center}
$\bea{l}
=
\displaystyle-\sum_{n\in \mathcal{I}{n_0}: k_n=k^{(1)}} \displaystyle \log(1-p^{(1)})
\vspace{0.0cm}\\ \hspace{0.3cm}
+\displaystyle\sum_{n\in \mathcal{I}{n_0}: k_n=k^{(2)}} \displaystyle \log(1-p^{(2)}) \vspace{0.0cm}\\ \hspace{0.3cm}
\displaystyle+\log u_{n_0}(k^{(2)})+\log p^{(2)}-I_{n_0}(k^{(2)},\sigma_{-n_0}) \vspace{0.0cm}\\ \hspace{0.3cm}
\displaystyle-\left(\log u_{n_0}(k^{(1)})+\log p^{(1)}-I_{n_0}(k^{(1)},\sigma_{-n_0})\right) \vspace{0.0cm}
\ena$
\end{center}
\begin{center}
$\bea{l}
=
\displaystyle
-\log(1-p^{(1)})\left|\mathcal{I}_{n_0}(k^{(1)})\right|
%
+\displaystyle\log(1-p^{(2)})\left|\mathcal{I}_{n_0}(k^{(2)})\right|  \vspace{0.0cm}\\ \hspace{0.3cm}
\displaystyle+\log u_{n_0}(k^{(2)})+\log p^{(2)}-I_{n_0}(k^{(2)},\sigma_{-n_0}) \vspace{0.0cm}\\ \hspace{0.3cm}
\displaystyle-\left(\log u_{n_0}(k^{(1)})+\log p^{(1)}-I_{n_0}(k^{(1)},\sigma_{-n_0})\right) \vspace{0.0cm}
\ena$
\end{center}
\begin{center}
$\bea{l}
=
\displaystyle
\Delta F_{n_0}\left(\sigma^{(1)}, \sigma^{(2)},\sigma_{-n_0}\right)\;, \vspace{0.0cm}
\ena$
\end{center}
where we used the facts that only users in $\mathcal{I}_{n_0}$ that transmit over channels $k^{(1)}$ and $k^{(2)}$ experience a change in their interference level, and the contributions of user $n_0$ to the log-interference experienced by its neighbors that transmit over channels $k^{(1)}$ and $k^{(2)}$ are $-\log(1-p^{(1)})$ and $-\log(1-p^{(2)})$, respectively.
Hence, (\ref{eq:ordinal_potential_def}) follows. Furthermore, $\phi(\sigma)$ is upper bounded as follows: $\phi(\sigma)<\sum_{n=1}^{N}{\max_{k}\log\left({u}_n(k)\right)}$. As a result, $\phi(\sigma)$ in (\ref{eq:exact_potential_function}) is a bounded exact potential function of the fairness game which completes the proof.
\QEDA

\bibliographystyle{ieeetr}

\end{document}

%% file: macros.tex
\newtheorem{theorem}{Theorem}
\newtheorem{lemma}{Lemma}
\newtheorem{remark}{Remark}
\newtheorem{definition}{Definition}
\newtheorem{corollary}{Corollary}
\newtheorem{proposition}{Proposition}

\newcommand{\beq}{\begin{equation}}
\newcommand{\eeq}{\end{equation}}
\newcommand{\bea}{\begin{array}}
\newcommand{\ena}{\end{array}}
\newcommand{\bds}{\begin {itemize}}
\newcommand{\eds}{\end {itemize}}
\newcommand{\bdf}{\begin{definition}}
\newcommand{\blm}{\begin{lemma}}
\newcommand{\edf}{\end{definition}}
\newcommand{\elm}{\end{lemma}}
\newcommand{\bthm}{\begin{theorem}}
\newcommand{\ethm}{\end{theorem}}
\newcommand{\bprp}{\begin{prop}}
\newcommand{\eprp}{\end{prop}}
\newcommand{\bcl}{\begin{claim}}
\newcommand{\ecl}{\end{claim}}
\newcommand{\bcr}{\begin{coro}}
\newcommand{\ecr}{\end{coro}}
\newcommand{\bquest}{\begin{question}}
\newcommand{\equest}{\end{question}}

\newcommand{\larrow}{{\larrow}}

\def\urltilda{\kern -.15em\lower .7ex\hbox{\~{}}\kern .04em}